\documentclass[]{article}  
\usepackage{amsfonts}
\usepackage{amssymb}
\usepackage{amsmath}
\usepackage{amsthm}
\usepackage{latexsym}
\usepackage[breaklinks]{hyperref}
\usepackage{fullpage} 
\usepackage{graphicx}
\usepackage{placeins}

\newcommand{\cd}{\mathrm{cd}}

\newcommand{\G}{\mathcal{G}}
\newcommand{\Vset}{\mathrm{VS}}
\newcommand{\Eset}{\mathrm{ES}}

\newcommand{\V}{\mathcal{V}}
\newcommand{\E}{\mathcal{E}}
\newcommand{\N}{\mathbb{N}}

\long\def\rem#1{}
\def\B{\{0,1\}}

\newcommand{\LG}{\mathcal{LG}}

\newtheorem{Definition}{Definition}
\newtheorem{theorem}{Theorem}
\newtheorem{lemma}[theorem]{Lemma}

\newcommand{\thmref}[1]{\hyperref[#1]{{Theorem~\ref*{#1}}}}
\newcommand{\lemref}[1]{\hyperref[#1]{{Lemma~\ref*{#1}}}}
\newcommand{\remref}[1]{\hyperref[#1]{{Remark~\ref*{#1}}}}
\newcommand{\corref}[1]{\hyperref[#1]{{Corollary~\ref*{#1}}}}
\newcommand{\eqnref}[1]{\hyperref[#1]{{Equation~(\ref*{#1})}}}
\newcommand{\claimref}[1]{\hyperref[#1]{{Claim~\ref*{#1}}}}
\newcommand{\remarkref}[1]{\hyperref[#1]{{Remark~\ref*{#1}}}}
\newcommand{\propref}[1]{\hyperref[#1]{{Proposition~\ref*{#1}}}}
\newcommand{\factref}[1]{\hyperref[#1]{{Fact~\ref*{#1}}}}
\newcommand{\defref}[1]{\hyperref[#1]{{Definition~\ref*{#1}}}}
\newcommand{\exampleref}[1]{\hyperref[#1]{{Example~\ref*{#1}}}}
\newcommand{\hypref}[1]{\hyperref[#1]{{Hypothesis~\ref*{#1}}}}
\newcommand{\secref}[1]{\hyperref[#1]{{Section~\ref*{#1}}}}
\newcommand{\chapref}[1]{\hyperref[#1]{{Chapter~\ref*{#1}}}}
\newcommand{\apref}[1]{\hyperref[#1]{{Appendix~\ref*{#1}}}}

\newcommand{\ignore}[1]{}
\begin{document}

\title{Improved Quantum Query Algorithms for\\  Triangle Finding and Associativity Testing%
\thanks{Partially supported by the French ANR Defis project
ANR-08-EMER-012 (QRAC)
and
the European Commission IST STREP project
25596 (QCS). Research at the Centre
for Quantum Technologies is funded by the Singapore Ministry of Education 
and the National Research Foundation.}
}
\author{Troy Lee\thanks{Centre for Quantum Technologies, National University
    of Singapore, Singapore 117543. \texttt{troyjlee@gmail.com}}
\and
Fr\'ed\'eric Magniez\thanks{CNRS, LIAFA, Univ Paris Diderot, Sorbonne Paris-Cit\'e, 75205 Paris, France. \texttt{frederic.magniez@univ-paris-diderot.fr}}
\and
Miklos Santha\thanks{CNRS, LIAFA, Univ Paris Diderot, Sorbonne Paris-Cit\'e, 75205 Paris, France;
and Centre for Quantum Technologies, National University
    of Singapore, Singapore 117543.
    \texttt{miklos.santha@liafa.univ-paris-diderot.fr}}
    }

\date{}
\maketitle

\begin{abstract}
We show that the quantum query complexity of detecting if an $n$-vertex graph contains a triangle is
$O(n^{9/7})$.  This improves the previous best algorithm of Belovs~\cite{spanCert} making 
$O(n^{35/27})$ queries.  For the problem of determining if an operation 
$\circ : S \times S \rightarrow S$ is associative, we give an algorithm making $O(|S|^{10/7})$ queries, 
the first improvement to the trivial $O(|S|^{3/2})$ application of Grover search.  

Our algorithms are designed using the learning graph framework of Belovs.  We give a family 
of algorithms for detecting constant-sized subgraphs, which can possibly be directed and colored.  
These algorithms are designed in a simple high-level language; our main theorem shows how this high-level 
language can be compiled as a learning graph and gives the resulting complexity.

The key idea to our improvements is to allow more freedom in the parameters of the database kept by the
algorithm.  As in our previous work~\cite{LMS11}, the edge slots maintained in the database are 
specified by a graph whose edges are the union of regular bipartite graphs,  the overall structure of which mimics 
that of the graph of the certificate.  By allowing these bipartite graphs to be unbalanced and of variable degree we 
obtain better algorithms.
\end{abstract}

\section{Introduction}
Quantum query complexity is a black-box model of quantum computation, where the resource measured is 
the number of queries to the input needed to compute a function.  This model captures the great algorithmic successes 
of quantum computing like the search algorithm of Grover~\cite{gro96} and the period finding subroutine of Shor's 
factoring algorithm~\cite{shor97}, while at the same time is simple enough that one can often show tight lower bounds.

Recently, there have been very exciting developments in quantum query complexity.  
Reichardt~\cite{Reichardt10advtight} showed that the general adversary bound, formerly just a lower bound technique 
for quantum query complexity~\cite{HoyerLeeSpalek07negativeadv}, is also an upper bound.  This characterization 
opens a new avenue for designing quantum query algorithms.  The general adversary bound can be written as a 
relatively simple semidefinite program, thus by providing a feasible solution to the minimization form of this 
program one can upper bound quantum query complexity.  

This plan turns out to be quite difficult to implement as the minimization form of the adversary bound has exponentially 
many constraints.  Even for simple functions it can be challenging to directly write down a feasible solution, much 
less worry about finding a solution with good objective value.  

To surmount this problem, Belovs~\cite{spanCert} introduced the beautiful model of learning graphs, which can be
viewed as the minimization form of the general adversary bound with additional structure imposed on the form 
of the solution.  This additional structure makes learning graphs easier to reason about by ensuring 
that the constraints are \emph{automatically} satisfied, leaving one to worry about optimizing the objective value.  

Learning graphs have already proven their worth, with Belovs using this model to give an algorithm for triangle 
finding with complexity $O(n^{35/27})$, improving the quantum walk algorithm~\cite{MSS07} of complexity 
$O(n^{1.3})$.  Belovs' algorithm was generalized to detecting constant-sized subgraphs~\cite{zhu11, LMS11}, giving 
an algorithm of complexity $o(n^{2-2/k})$ for determining if a graph contains a $k$-vertex subgraph $H$, 
again improving the~\cite{MSS07} bound of $O(n^{2-2/k})$.  All these 
algorithms use the most basic model of learning graphs, that we also use in this paper.  A more general model of 
learning graphs (introduced, though not used in Belovs' original paper) was used to give an $o(n^{3/4})$ algorithm for 
$k$-element distinctness, when the inputs are promised to be of a certain form~\cite{BL11}.  Recently, Belovs 
further generalized the learning graph model and removed this promise to obtain an $o(n^{3/4})$ algorithm for the 
general $k$-distinctness problem~\cite{Belovs12}.

In this paper, we continue to show the power of the learning graph model.  We give an algorithm for detecting 
a triangle in a graph making $O(n^{9/7})$ queries.  This lowers the exponent of Belovs algorithm from 
about $1.296$ to under $1.286$.  For the problem of testing if an operation $\circ: S \times S \rightarrow S$ is 
associative, where $|S|=n$, we give an algorithm making $O(n^{10/7})$ queries, the first improvement over the 
trivial application of Grover search making $O(n^{3/2})$ queries.  Previously, D\"{o}rn and Thierauf \cite{DT08}
gave a quantum walk based algorithm to test if $\circ: S \times S \rightarrow S'$ is associative that improved on Grover 
search but only when $|S'| < n^{3/4}$.

More generally, we give a family of algorithms for detecting constant-sized subgraphs, which can possibly be 
directed and colored.  Algorithms in this family can 
be designed using a simple high-level language.  Our main theorem shows how to compile this language as 
a learning graph, and gives the resulting complexity.  We now explain in more detail how our algorithms 
improve over previous work.

{\bf Our contribution.}
We will explain the new ideas in our algorithm using triangle detection as an example.  We first 
review the quantum walk algorithm of~\cite{MSS07}, and the learning graph algorithm of Belovs~\cite{spanCert}.  
For this high-level overview we just focus on the database of edge slots of the input graph $G$ 
that is maintained by the algorithm.  A quantum walk algorithm explicitly maintains such a database, and the 
nodes of a learning graph are labeled by sets of queries which we will similarly interpret 
as the database of the algorithm.  

In the quantum walk algorithm~\cite{MSS07} the database consists of an $r$-element subset of the $n$-vertices of~$G$ 
and all the edge slots among these $r$-vertices.  That is, the presence or absence of an edge in $G$ among a 
complete $r$-element subgraph is maintained by the database.  In the learning graph algorithm 
of Belovs, the database consists of a random subgraph with edge density $0\le s \le 1$ of a complete 
$r$-element subgraph. 
In this way, on average, $O(sr^2)$ many edge slots are queried among the $r$-element subset, making 
it cheaper to set up this database.  This saving is what results in the improvement of Belovs' algorithm.  
Both algorithms finish by using search plus graph collision to locate a vertex that 
is connected to the endpoints of an edge present in the database, forming a triangle.

Zhu~\cite{zhu11} and Lee et al.~\cite{LMS11} extended the triangle finding algorithm of Belovs to finding constant 
sized subgraphs.  While the algorithm of Zhu again maintains a database of a random subgraph of an $r$-vertex 
complete graph with edge density $s$, the algorithm of Lee et al.\  instead used a more structured database.  
Let $H$ be a $k$-vertex subgraph with vertices labeled from $[k]$.  To determine if $G$ contains a copy of $H$, 
the database of the algorithm consists of $k-1$ sets $A_1, \ldots, A_{k-1}$ of size $r$ and 
for every $\{i,j\} \in H -\{k\}$ the edge slots of~$G$ according to a 
$sr$-regular bipartite graph between $A_i$ and $A_j$.  Again both algorithms finish by using search plus graph 
collision to find a vertex connected to edges in the database to form a copy of $H$.

In this work, our database is again the edge slots of~$G$ queried according according to the union 
of regular bipartite graphs whose overall structure mimics the structure of $H$.  
Now, however, we allow optimization over all parameters of the database---we allow the size of the set $A_i$ to be a 
parameter $r_i$ that can be independently chosen; similarly, we allow the degree of the bipartite 
graph between $A_i$ and $A_j$ to be a variable $d_{ij}$.  This greater freedom in the parameters of the 
database allows the improvement in triangle finding from $O(n^{35/27})$ to $O(n^{9/7})$.  Instead of 
an $r$-vertex graph with edge density $s$, our algorithm uses as a database a complete unbalanced bipartite 
graph with left hand side of size $r_1$ and right hand side of size $r_2$.  Taking $r_1 < r_2$ allows a more 
efficient distribution of resources over the course of the algorithm.  As before, the algorithm finishes by using search 
plus graph collision to find a vertex connected to endpoints of an edge in the database.  

The extension to functions of the form $f:[q]^{n \times n} \rightarrow \B$, like associativity, 
comes from the fact that the basic learning graph model that we use depends only on the structure of a 1-certificate 
and not on the values in a  1-certificate.  
This property means that an algorithm for detecting a subgraph $H$ can be immediately 
applied to detecting $H$ with specified edge colors in a colored graph.   

If an operation $\circ: S \times S \rightarrow S$ is non-associative, then there are 
elements $a,b,c$ such that $a \circ (b \circ c) \ne (a \circ b) \circ c$.  A certificate consists of the $4$ (colored and 
directed) edges $b\circ c = e, a \circ e, a \circ b =d$, and $d \circ c$ such that $a \circ e \ne d \circ c$.  The graph 
of this certificate is a $4$-path with directed edges, and using our algorithm for this graph gives complexity 
$O(|S|^{10/7})$.

We provide a high-level language for designing algorithms within our framework.  The algorithm begins by choosing 
size parameters for each $A_i$ and degree parameters for the bipartite graph 
between $A_i$ and $A_j$.  Then one can choose the order in which to load vertices $a_i$ and edges 
$(a_i, a_j)$ of a $1$-certificate, according to the rules that both endpoints of an edge must be loaded before the edge, 
and at the end all edges of the certificate must be loaded.  Our main theorem \thmref{thm:main} shows how 
to implement this high-level algorithm as a learning graph and gives the resulting complexity.  

With larger subgraphs, optimizing over the set size and degree parameters to obtain an algorithm of minimal 
complexity becomes unwieldy to do by hand.  Fortunately, this can be phrased as a linear program and we 
provide code to compute a set of optimal parameters\footnote{code is available at 
\url{https://github.com/troyjlee/learning_graph_lp}}.

\section{Preliminaries}\label{sec:lg}
The \emph{quantum query complexity} of a function $f$, denoted $Q(f)$, is the number of 
input queries needed to evaluate $f$ with error at most $1/3$.  We refer 
the reader to the survey~\cite{hs05} for precise definitions and background.

For any integer $q\geq 1$, let $[q] = \{1, 2,\ldots, q\}$.
We will deal with boolean functions  of the form 
$f : [q]^{n \times n} \rightarrow \B$, where the input to the function can 
be thought of as the complete directed graph (possibly with self-loops) on vertex set $[n]$, whose edges are colored by elements 
from $[q]$. When $q=2$, the input is of course just a directed graph (again possibly with self-loops).
A {\em partial assignment} is an element of the set $([q] \cup \{\star\})^{n \times n}$. For  partial assignments
$\alpha_1$ and $\alpha_2$ we say that $\alpha_1$ is a {\em restriction} of $\alpha_2$ (or alternately $\alpha_2$ 
is an {\em extension} of $\alpha_1$) if whenever $\alpha_1(i,j) \neq \star$ then $\alpha_1(i,j) = \alpha_2(i,j)$.
A \emph{$1$-certificate} for $f$ is a partial assignment $\alpha$ such that $f(x)=1$ for every
extension $x \in [q]^{n \times n}$ of $\alpha$. 
If $\alpha$ is a 1-certificate and $x \in [q]^{n \times n}$ is an extension of 
$\alpha$, we also say that $\alpha$ is a 1-certificate for $f$ and $x$.
A 1-certificate $\alpha$ is {\em minimal} if no proper restriction of $\alpha$ is a 1-certificate.
The {\em index set} of a 1-certificate $\alpha$ for $f$ is the set
$I_{\alpha} = \{ (i,j) \in [n] \times [n] : \alpha(i,j) \neq \star \}$. 
Besides these standard notions, we will also need the notion of the graph of a $1$-certificate.
For a graph $G$, let $V(G)$ denote the set of vertices, and $E(G)$ the set of edges of~$G$.

\begin{Definition}[Certificate graph]
Let $\alpha$ be a $1$-certificate for $f : [q]^{n \times n} \rightarrow \B$.  The \emph{certificate graph} $H_{\alpha}$ of $\alpha$ is defined
by $E (H_{\alpha}) = I_{\alpha}$, and $V(H_{\alpha})$ is the set of elements in $[n]$ which are adjacent to an edge in $I_{\alpha}$.
The \emph{size} of a certificate graph is the cardinality of its edges. 
A \emph{minimal certificate graph} for $x$, such that $f(x)=1$, is the certificate graph of a minimal 1-certificate for $f$ and $x$.
The \emph{$1$-certificate complexity} of $f$ is the size of the biggest minimal certificate graph for some $x$ such that $f(x)=1$.
\end{Definition}
Intuitively, if $x \in [q]^{n \times n}$ is an extension of a $1$-certificate $\alpha$, the certificate graph of $\alpha$ 
represents queries that are sufficient to verify $f(x)=1$.

Vertices of our learning graphs will be labeled
by sets of edges coming from the union of a bunch of bipartite graphs.  
We will specify these bipartite graphs by their degree sequences, the number of vertices on 
the left hand side and right hand side of a given degree.  The following notation will be useful to do this.
\begin{Definition}[Type of bipartite graph]
\label{def:type}
A bipartite graph between two sets $Y_1$ and $Y_2$ is of \emph{type}
$(\{(n_1,d_1), \ldots, (n_j,d_j)\} , \{(m_1,g_1), \ldots, (m_\ell,g_\ell)\})$ if
$Y_1$ has $n_i$ vertices of degree $d_i$ for $i = 1 , \ldots , j$, and 
$Y_2$ has $m_i$ vertices of degree $g_i$ for $i = 1 , \ldots , \ell$, and this is a complete 
listing of vertices in the graph, i.e. $|Y_1|=\sum_{i=1}^j n_i$ 
and $|Y_2|=\sum_{i=1}^\ell m_i$.  Note also that $\sum_{i=1}^j n_i d_i = \sum_{i=1}^\ell m_i g_i$.  
\end{Definition}

\paragraph{Learning graphs}
We now formally define a learning graph and its complexity.  We first define a learning graph in the abstract.
\begin{Definition}[Learning graph]\ 
A learning graph $\G$ is a 5-tuple $(\V, \E, w, \ell, \{p_y: y \in Y\})$
where $(\V,\E)$ is a  
rooted, weighted and directed acyclic graph,  the weight function 
$w : \E \rightarrow \mathbb{R}$ maps learning graph 
edges to positive real numbers, the length function $\ell:\E \rightarrow \N$ assigns each edge a 
natural number, and $p_y : \E \rightarrow \mathbb{R}$   
is a unit flow whose source is the root, for every $y \in Y$.
\end{Definition}

A learning graph for a function has additional requirements as follows.
\begin{Definition}[Learning\! graph\! for\! a\! function]
Let $f : [q]^{n \times n} \rightarrow \B$ be a function.  A learning graph $\G$ for $f$ is a 5-tuple 
$(\V, \E,S, w, \{p_y : y \in f^{-1}(1)\})$, where 
$S: \V \rightarrow 2^{n \times n}$ maps $v \in \V$ to a label $S(v)\subseteq [n] \times [n]$ of variable indices,
and $(\V, \E, w, \ell, \{p_y: y \in f^{-1}(1)\})$ is a learning graph 
for the length function $\ell$ defined as $\ell((u,v)=|S(v) \setminus S(u)|$ for each edge $(u,v)$.
For the root $r \in \V$ we have $S(r) = \emptyset$, 
and every learning graph edge $e=(u,v)$ satisfies $S(u) \subseteq S(v)$.
For each input $y \in f^{-1}(1)$, the set $S(v)$ contains the index set of a $1$-certificate for $y$ on $f$, for every sink 
$v \in \V$ of $p_y$. 
\end{Definition}
In our construction of learning graphs we usually define $S$ by more colloquially stating the {\em label} of each vertex.  
Note that it can be 
the case for an edge $(u,v)$ that $S(u)=S(v)$ and the length of the edge is zero.
In Belovs~\cite{spanCert} what we define here is called a reduced learning graph, and a learning graph is 
restricted to have all edges of length at most one.

In this paper we will discuss functions whose inputs are themselves graphs.  To prevent confusion
we will refer to vertices and edges of the learning graph as \emph{$L$-vertices} and \emph{$L$-edges} respectively.

We now define the complexity of a learning graph.  For the analysis it will be helpful to define the complexity not just for
the entire learning graph but also for {\em stages} of the learning graph $\G$.  
By  {\em level $d$} of $\G$ we refer to the set of vertices at distance 
$d$ from the root.  A {\em stage} is the set of edges of $\G$ between level $i$ and level $j$, for some $i<j$.  
For a subset $V \subseteq \V$ of the $L$-vertices let $V^+ = \{ (v,w) \in \E : v \in V\}$ and similarly let 
$V^- = \{ (u,v) \in\E : v \in V\}$.  For a vertex $v$ we will write $v^+$ instead of ${\{v\}}^+$, and similarly for $v^-$ 
instead of ${\{v\}}^-$.

\begin{Definition}[Learning graph complexity] \label{e:lg_cost}
Let $\G$ be a learning graph, and let $E \subseteq \E$ be the edges of a stage.
The negative complexity of $E$ is $$C_0(E)=\sum_{e \in E} \ell(e) w(e).$$ The positive complexity
of $E$ under the flow $p_y$ is 
\begin{equation*}
C_{1,y}(E)= 
\sum_{e \in E}  \frac{\ell(e)}{w(e)}p_y(e)^2.
\end{equation*}
The positive complexity of $E$ is $$C_1(E) =\max_{y \in Y} C_{1,y}(E).$$
The complexity of $E$ is $C(E)=\sqrt{C_0(E) C_1(E)}$, and 
the learning graph complexity of $\G$ is $C(\G)= C(\E)$.  
The learning graph 
complexity of a function $f$, denoted $\LG(f)$, is the minimum learning graph complexity of 
a learning graph for $f$.
\end{Definition}


\begin{theorem}[Belovs] \label{t:lg_alg} 
$Q(f)=O(\LG(f))$.  
\end{theorem}
Originally Belovs showed this theorem with an additional $\log q$ factor for functions over an input alphabet of 
size $q$; this logarithmic factor was removed in~\cite{BL11}.

\paragraph{Analysis of learning graphs}
Given a learning graph $\G$, the easiest way to obtain another learning graph is to modify the
weight function of $\G$. We will often use this reweighting scheme to obtain learning graphs with better complexity 
or complexity that is more convenient to analyze.
When $\G$ is understood from the context, and when $w'$ is the new weight function, for the edges $E \subseteq \E$
of a stage, we denote the complexity of $E$ with respect to $w'$ by $C^{w'}(E)$.

The following useful lemma of Belovs gives an example of the reweighting method.  It 
shows how to upper bound the complexity of a learning graph by partitioning it into a 
constant number of stages and summing the complexities of the stages.  

\begin{lemma}[Belovs]\label{l:stages}
If $\E$ can be partitioned into a constant number $k$ of stages $E_1, \ldots, E_k$, 
then there exists a weight function $w'$ such that
$$C^{w'}(\G)=O(C(E_1) + \ldots + C(E_k)).$$
\end{lemma}

Now we will focus on evaluating the complexity of a stage.  Our learning graph algorithm for triangle 
detection is of a very simple form, where all $L$-edges present in the graph have weight one, all $L$-vertices in a level 
have the same degree, incoming and outgoing flows are uniform over a subset of $L$-vertices in each level, and 
all $L$-edges between levels are of the same length.
In this case the complexity of a stage between consecutive levels can be estimated quite 
simply.
\begin{lemma}\label{e:simple_cost}
Consider a stage of a learning graph between consecutive levels.
Let $V$ be the set of $L$-vertices at the beginning of the stage.  
Suppose that each $L$-vertex $v \in V$ is of degree-$d$ with all outgoing $L$-edges $e$
of weight $w(e)=1$ and of length $\ell(e)\le\ell$.
Furthermore, say that the incoming flow is uniform over $L$-vertices 
$W \subseteq V$, and is uniformly directed from each $L$-vertex $v \in W$ 
to $g$ of the $d$ possible neighbors.  
Then the complexity of this stage is at most
$\ell \sqrt{\frac{d|V|}{g|W|}}$.
\end{lemma}

\begin{proof}
The total weight is $d|V|$.  The flow through each of the $g |W|$ many $L$-edges is 
$(g |W|)^{-1}$.  Plugging these into \defref{e:lg_cost} gives the lemma.
\end{proof}
To analyze the cost of our algorithm for triangle detection, we will repeatedly use \lemref{e:simple_cost}.  
The contributions to the complexity of a stage
are naturally broken into three parts: the 
length $\ell$, the \emph{vertex ratio} $|V| / |W|$, and the \emph{degree ratio} $d/g$.  This terminology will 
be  helpful in discussing the complexity of stages.  

For our more general framework given in \secref{sec:general}, flows will no longer be uniform.  To evaluate the 
complexity in this case, we will use several lemmas developed in~\cite{LMS11}.  The main idea is to use 
the symmetry of the function to decompose flows as a convex combination of uniform flows over disjoint edge sets.  
A natural extension of \lemref{e:simple_cost} can then be used to evaluate the complexity.  To state the lemma we first 
need a definition.
For a set of $L$-edges $E$, we let $p_y(E)$ denote the {\em value} of the flow $p_y$ over $E$, that is 
$p_y(E) = \sum_{e\in E}p_y(e)$. 
\begin{Definition}[Consistent flows] 
Let $E$ be a stage of $\G$ between two consecutive levels, 
and let $V_1, \ldots, V_s$ be a partition of
the $L$-vertices at the beginning of the stage.
We say that $\{p_y\}$ is consistent with $V_1^+ , \ldots, V_s^+ $ if  
$p_y(V_i^{+} )$ is independent of $y$ for each $i$.
\end{Definition}

The next lemma is the main tool for evaluating the complexity of learning graphs in our main theorem, \thmref{thm:main}.

\begin{lemma}[\cite{LMS11}]
\label{e:simple_costg}
Let $E$ be a stage of $\G$ between two consecutive levels.  Let $V$ be the set of $L$-vertices at the 
beginning of the stage and suppose that each $v \in V$ has outdegree $d$ and all $L$-edges $e$ of the stage 
satisfy $w(e)=1$ and $\ell(e) \le \ell$.  
Let $V_1, \ldots, V_s$ be a partition of~$V$, 
and for all $y$ and $i$,
let $W_{y,i} \subseteq V_i$ be the set of vertices in $V_i$
which receive positive flow under $p_y$.
Suppose that 
\begin{enumerate}
\item the flows $\{p_y\}$ are consistent with $\{{V_i}^+\}$,
\item  
$|W_{y,i}| $ is independent from $y$ for every $i$, and for all $v \in W_{y,i}$ we have $p_y(v^+) = p_y({V_i}^+)/|W_{y,i}|$,
\item there is a $g$ such that for each vertex $v \in W_{y,i}$ the flow is directed uniformly to $g$ of the $d$ many 
neighbors.  
\end{enumerate}
Then there is a new weight function $w'$ such that
\begin{equation}
\label{e:cost_formula}
C^{w'}(E) \leq   \max_i  \ell \sqrt{\frac{d}{g} \frac{ |V_i|}{|W_{y,i}|}} \enspace .
\end{equation}
\end{lemma}

We will refer to $\max_i |V_i|/|W_{y,i}|$ as the {\em maximum vertex ratio}.
For the most part we will deal with the problem of detecting a (possibly directed and colored) subgraph in an $n$-
vertex graph.  We will be interested in symmetries induced by permuting the
elements of $[n]$, as such permutations do not change the property of containing a fixed subgraph.
We now state two additional lemmas from~\cite{LMS11} that use this symmetry to help establish the hypotheses of \lemref{e:simple_costg}.

For $\sigma \in S_n$, we define and also denote by $\sigma$
the permutation over $[n] \times [n]$ such that $\sigma (i,j) = (\sigma(i), \sigma(j) )$.  Recall that each $L$-vertex $u$ is labeled by a $k$-partite graph on 
$[n]$, say with color classes $A_1, \ldots, A_k$, and that we identify an $L$-vertex with its label.  For 
$\sigma \in S_n$ we define the action of $\sigma$ on $u$ as $\sigma(u) = v$, where $v$ is a $k$-partite graph 
with color classes $\sigma(A_1), \ldots, \sigma(A_k)$ and edges $\{\sigma(i), \sigma(j)\}$ for every edge $\{i,j\}$ in $u$.  

Define an equivalence class $[u]$ of $L$-vertices by $[u] = \{ \sigma(u): \sigma \in S_n\}$.  We say that $S_n$ {\em acts 
transitively} on flows $\{p_y\}$ if for every $y, y'$ there is a $\tau \in S_n$ such that $p_y((u,v))=p_{y'}((\tau(u), \tau(v))$ 
for all $L$-edges $(u,v)$.

The following lemma from~\cite{LMS11} shows that if $S_n$ acts transitively on a set of flows $\{p_y\}$ then they are 
consistent with 
${[v]}^+$, where $v$ is a vertex at the beginning of a stage between consecutive levels. 
This will set us up to satisfy hypothesis~(1) of \lemref{e:simple_costg}.
\begin{lemma}[\cite{LMS11}]
\label{lem:trans_cons}
Consider a learning graph $\G$ and a set of flows $\{p_y\}$ such that $S_n$ acts transitively on $\{p_y\}$.  
Let $V$ be the set of $L$-vertices of $\G$ at some given level.
Then $\{p_y\}$ is consistent with $\{[u]^+ : u \in V\}$, and, similarly, $\{p_y\}$ is consistent with $\{[u]^- : u \in V\}$.
\end{lemma}

The next lemma gives a sufficient condition for hypothesis~(2) of \lemref{e:simple_costg} to be satisfied.  
The partition of vertices in \lemref{e:simple_costg} will be taken according to the equivalence classes $[u]$. 
 
\begin{lemma}[\cite{LMS11}]
\label{e:symmetry}
Consider a learning graph and a set of flows $\{p_y\}$ such that $S_n$ acts transitively on $\{p_y\}$.  
Suppose that for every $L$-vertex $u$ and flow $p_y$ such that $p_y(u^{-}) > 0$,
\begin{enumerate}
  \item 
   the flow from $u$ is uniformly directed to $g^+([u])$ many neighbors,
  \label{esh2}
  \item 
  for every $L$-vertex $w$, 
  the number of incoming edges with from $[w]$ to $u$ is $g^-([w], [u])$.  
  \label{esh3}
\end{enumerate}
Then for every $L$-vertex $u$ the flow entering $[u]$ is uniformly distributed over $W_{y,[u]} \subseteq [u]$ where 
  $|W_{y, [u]}|$ is independent of $y$.
\end{lemma}


\section{Triangle algorithm}
\label{sec:triangle}

\begin{figure}
\centerline{\includegraphics[scale=0.33]{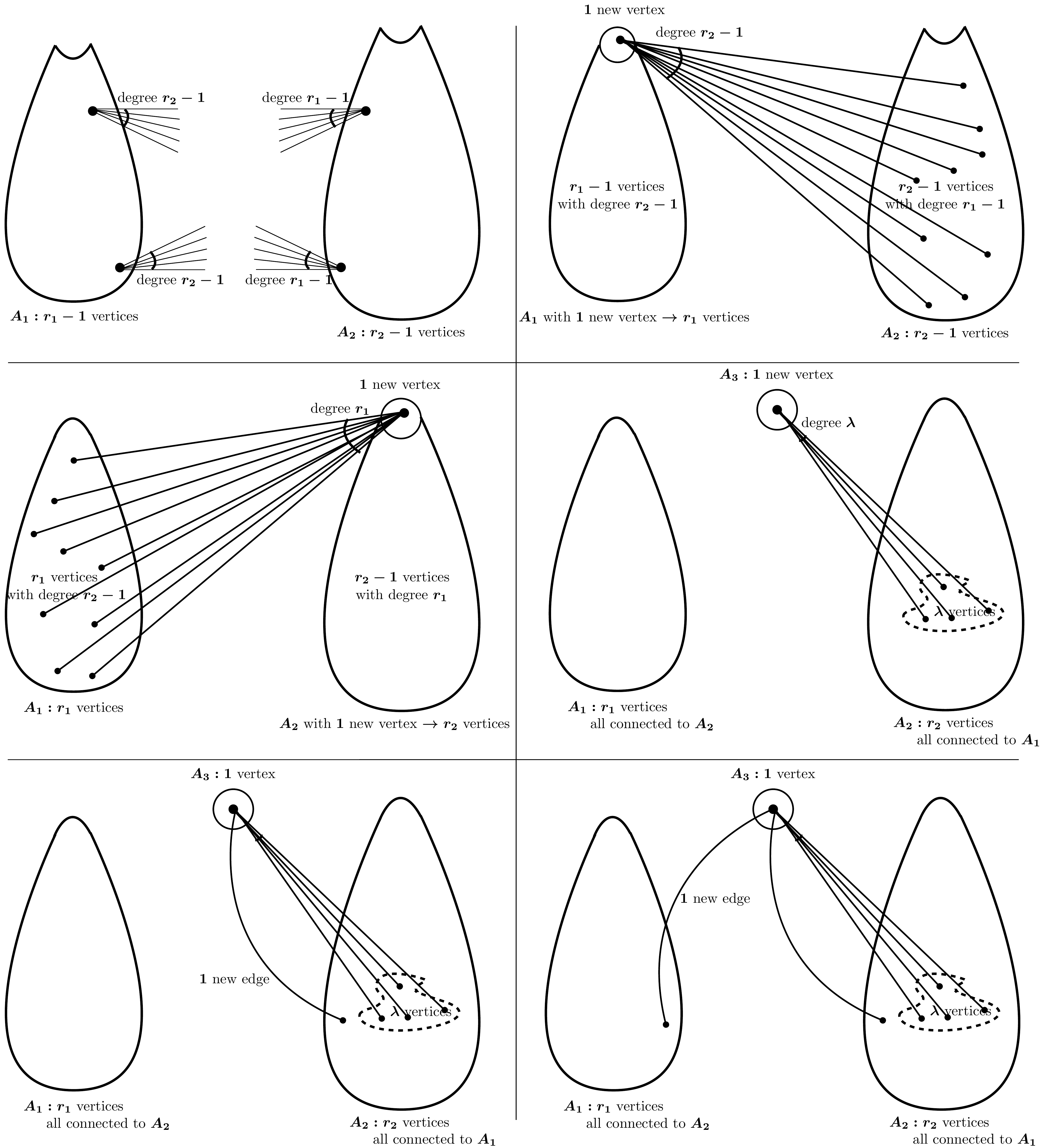}}
\caption{Stages 1-6 
for Triangle Algorithm\label{fig:triangle}} 
\end{figure}

\begin{theorem}\label{thm:triangle}
There is a bounded-error quantum query algorithm for detecting if an $n$-vertex graph contains a triangle making 
$O(n^{9/7})$ many queries.
\end{theorem}
\begin{proof}
We will show the theorem by giving a learning graph of the claimed complexity, which is sufficient by 
\thmref{t:lg_alg}.  We will define the learning graph 
by stages; let $V_t$ denote the $L$-vertices of the learning graph present at the beginning of stage $t$.  The $L$-edges 
between $V_t$ and $V_{t+1}$ are defined in the obvious way---there is an $L$-edge between $v_t \in V_t$ and 
$v_{t+1} \in V_{t+1}$ if the graph labeling $v_t$ is a subgraph of the graph labeling $v_{t+1}$, and all such 
$L$-edges have weight one.  
The root of the learning graph is labeled by the empty graph.

For a positive input graph $G$, let $a_1, a_2, a_3$ be the vertices of a triangle of $G$.  The algorithm 
(see Figure~\ref{fig:triangle})
depends on set size parameters $r_1, r_2 \in [n], r_1, r_2 = o(n)$ and a vertex degree parameter 
$\lambda \in [n]$ that will be optimized later.  We will choose $r_1 < r_2$ such that $r_2/ r_1$ is an integer.  
The cost of each stage will be upper bounded using \lemref{e:simple_cost}.

\paragraph{Stage 1 (Setup):}
The initial level $V_1$ consists 
of the root of the learning graph labeled by the empty graph. 
The level $V_2$ consists of all $L$-vertices labeled by a complete unbalanced 
bipartite graph with disjoint color classes $A_1, A_2 \subseteq [n]$ where $|A_1|= r_1-1$ and $|A_2| = r_2-1$ and 
$r_1 \le r_2$.   
Flow is uniform from the root 
to all $L$-vertices such that 
$a_i \not \in A_1, a_i \not \in A_2$ for $i=1,2,3$. \\
\textbf{Cost:} The hypotheses of \lemref{e:simple_cost} hold trivially at this stage.  
The length of this stage is $O(r_1 r_2)$.  The vertex ratio is $1$, and the degree ratio is  
$$\frac{{n \choose r_1}{ n- r_1 \choose r_2}}{{n -3 \choose r_1}{n-r_1 -3 \choose r_2}} = O(1),$$
as $r_1, r_2 = o(n)$.  Thus the overall cost is $O(r_1 r_2)$.

\paragraph{Stage 2 (Load $a_1$):} During this stage we add a vertex to the set $A_1$ and connect it to all 
vertices in $A_2$.  Formally,
$V_3$ consists of all vertices labeled by a complete bipartite graph between color classes $A_1, A_2$ of sizes 
$r_1, r_2 - 1$, respectively. The flow goes uniformly to those $L$-vertices where $a_1$ is the vertex added to $A_1$.  \\
\textbf{Cost:}  By the definition of stage 1, the flow is uniform over $L$-vertices at the beginning of stage 2.  
The out-degree of every $L$-vertex in $V_1$ is $n-r_1-r_2+2$.  Of these, in $L$-vertices with flow, exactly one 
edge is taken by the flow.  Thus we can apply \lemref{e:simple_cost}.
Since the degree ratio was $O(1)$ for the first stage,
 the vertex ratio is also $O(1)$ for this stage.  The length is $r_2-1$.  The degree ratio is $O(n)$.
Thus the cost of this stage is $O(\sqrt{n} r_2)$.

\paragraph{Stage 3 (Load $a_2$):} We add a vertex to $A_2$ and connect it to all of the $r_1$ many vertices 
in $A_1$.  Thus the $L$-vertices at the end of stage $3$ consist of all complete bipartite graphs between sets 
$A_1, A_2$ of sizes $r_1, r_2$, respectively.
The flow goes uniformly to those $L$-vertices where $a_2$ is added at this stage to $A_2$.  Note that since we 
work with a complete bipartite graph, 
if $a_1 \in A_1$ and $a_2 \in A_2$ then the edge $\{a_1, a_2\}$ is automatically present.\\
\textbf{Cost:} The amount of flow in a vertex with flow at the beginning of stage 3 is the same as at the beginning 
of stage 2, as the flow out-degree in stage 2 was one and there was no merging of flow.  Thus flow is still uniform 
at the beginning of stage 3.  The out-degree of each $L$-vertex is $n-r_1-r_2+1$ and again for $L$-vertices with flow, 
the flow out-degree is exactly one.  Thus we can again apply \lemref{e:simple_cost}.

The length of this stage is $r_1$.  The vertex ratio is $O(n/r_1)$ as flow is present in $L$-vertices where $a_1$ is in the 
set $A_1$ of size $r_1$ (and such that $a_2, a_3$ are not loaded which only affects 
things by a $O(1)$ factor).  The degree ratio is again $O(n)$ as the flow only uses 
$L$-edges where $a_2$ is added out of $n-r_1-r_2+1$ possible choices.  Thus the cost of this stage is 
$O(\sqrt{n/r_1} \sqrt{n} r_1) = O(n \sqrt{r_1})$.

\paragraph{Stage 4 (Load $a_3$):}
We pick a vertex $v$ and $\lambda$ many edges connecting $v$ to $A_2$.  Thus the $L$-vertices at the 
end of stage~4 are labeled by edges that are the union of two bipartite graphs: a complete bipartite graph 
between $A_1, A_2$ of sizes $r_1, r_2$, and a bipartite graph between $v$ and $A_2$ of type 
$\{(1, \lambda)\}, \{(\lambda, 1), (r_2 - \lambda, 0)\}$.  
Flow goes uniformly to those $L$-vertices where $v=a_3$ and the edge $\{a_2, a_3\}$ is not loaded.\\
\textbf{Cost:}  Again the amount of flow in a vertex with flow at the beginning of stage 4 is the same as 
at the beginning of stage~3, as the flow out-degree in stage 3 was one and there was no merging of flow.  Thus 
the flow is still uniform.  The out-degree of $L$-vertices is $(n-r_1-r_2) {r_2 \choose \lambda}$, and the flow out-degree 
is ${r_2-1 \choose \lambda}$.  Thus we can again apply \lemref{e:simple_cost}.

The length of this stage is $\lambda$.  At the beginning of stage 4 flow is present in those 
$L$-vertices where $a_1 \in A_1, a_2 \in A_2$ and $a_3$ is not loaded.  Thus the vertex ratio is 
$O((n/r_1)(n/r_2))$.  Finally, the degree ratio is $O(n)$.  Thus the overall cost of this stage is
$$
O\left(\sqrt{\frac{n}{r_1}}\sqrt{\frac{n}{r_2}} \sqrt{n} \lambda\right) = O\left(\frac{n^{3/2} \lambda}{\sqrt{r_1 r_2}}\right).
$$

\paragraph{Stage 5 (Load $\{a_2,a_3\}$):} We add one new edge between $v$ and $A_2$.  Thus the $L$-vertices 
at the end of this stage will be labeled by the union of edges in two bipartite graphs:
a complete bipartite graph between $A_1, A_2$ of sizes $r_1, r_2$, and the second between 
$v$ and $A_2$ of type $\{(1, \lambda+1)\}, \{(\lambda+1, 1), (r_2 - \lambda-1, 0)\}$.  
Flow goes uniformly along those $L$-edges where the edge added is $\{a_2, a_3\}$.\\
\textbf{Cost:} The flow is uniform at the beginning of this stage, as it was uniform at the beginning of 
stage 4, the flow out-degree was constant in stage 4, and there was no merging of flow.  Each $L$-vertex has 
out-degree $r_2 - \lambda$ and the flow-outdegree is one.  Thus we can again apply \lemref{e:simple_cost}.

The length of this stage is one.  The vertex ratio is $O((n/r_1)(n/r_2)n)$ as flow is present in a 
constant fraction of those $L$-vertices where $a_1 \in A_1, a_2 \in A_2$ and $v=a_3$.  
The degree ratio is $r_2-\lambda$, as there are this many possible edges to add and the flow uses one.  
Thus the overall cost of this stage is 
$$
O\left(\sqrt{\frac{n}{r_1}}\sqrt{\frac{n}{r_2}} \sqrt{n} \sqrt{r_2}\right) =O\left(\frac{n^{3/2}}{\sqrt{r_1}}\right).
$$

\paragraph{Stage 6 (Load $\{a_1, a_3\}$):}
We add one new edge between $v$ and $A_1$.  Thus the $L$-vertices at the end of this stage will be labeled by 
the union of three bipartite graphs between $A_1, A_2$ and $v, A_2$ as before, and additionally between 
$v, A_1$ of type $\{(1,1)\}, \{(1,1), (r_1-1, 0)\}$.    Flow goes uniformly on those $L$-edges where $\{a_1, a_3\}$ is 
added.\\
\textbf{Cost:} Again flow is uniform as it was at the beginning of stage 5, the flow out-degree was constant and there was 
no merging.  Each $L$-vertex has out degree $r_1$ and the flow out-degree is one.  Thus we can again apply \lemref{e:simple_cost}.

The length of this stage is one. The vertex ratio is $O((n/r_1)(n/r_2)n(r_2/\lambda))$ 
as flow is present in a constant fraction 
of those $L$-vertices where $a_1 \in A_1, a_2 \in A_2, v=a_3$ and $\{a_2, a_3\}$ is present.  The degree ratio is 
$r_1$.  Thus the overall cost of this stage is
$$
O\left(\sqrt{\frac{n}{r_1}}\sqrt{\frac{n}{r_2}} \sqrt{n}\sqrt{\frac{r_2}{\lambda}} \sqrt{r_1}\right)= 
O\left(\frac{n^{3/2}}{\sqrt{\lambda}}\right).
$$

By choosing $r_1 = n^{4/7}, r_2 = n^{5/7}, \lambda = n^{3/7}$ we can make all costs, and thus their sum, 
$O(n^{9/7})$. 
\end{proof}

To quickly compute the stage costs, it is useful 
to associate to each stage a {\em local cost} and {\em global cost}.  The local cost is the product of the square root of 
the degree ratio and the length of a stage.  The global cost is the square root of the factor by which the stage 
increases the vertex ratio---we call this a global cost as it is propagated from one stage to the next.  Thus the 
square root of the vertex ratio at stage $t$ will be given by the product of the global costs of stages $1, \ldots, t-1$.  
As the cost of each stage is the product of the 
square root of the vertex ratio, square root of the degree ratio, and length, it can be computed by multiplying the local 
cost of the stage with the product of the global costs of all previous stages.
\begin{center}
\begin{tabular}{|r|c c c c c c|}
\hline
Stage & 1 & 2 & 3 & 4 & 5 & 6\\
\hline
Global cost & $1$ & $\sqrt{{n}/{r_1}}$ & $\sqrt{{n}/{r_2}}$  &
$\sqrt{n}$ & $\sqrt{{r_2}/{\lambda}}$ & \\
Local cost & $r_1 r_2$ & $\sqrt{n} r_2$ & $\sqrt{n} r_1$ & 
$\sqrt{n}\lambda$ & $\sqrt{r_2}$  & $\sqrt{r_1}$ \\
\hline
Cost & $r_1 r_2$ & $\sqrt{n} r_2$ & $n \sqrt{r_1}$ &
$n^{3/2} \lambda / \sqrt{r_1 r_2}$ & $n^{3/2} / \sqrt{r_1}$ & $n^{3/2}/\sqrt{\lambda}$ \\
\hline
Value & $n^{9/7}$ & $n^{17/14}$ & $n^{9/7}$  & $n^{9/7}$ & $n^{17/14}$ & $n^{9/7}$ \\
\hline
\end{tabular}
\end{center}

\section{An abstract language for learning graphs}
\label{sec:general}
 In this section we develop a high-level language for designing algorithms to detect constant-sized subgraphs, and 
 more generally to compute functions $f: [q]^{n \times n} \rightarrow \B$ with constant-sized $1$-certificate complexity. 
 This high-level language consists of commands like ``load a vertex'' or ``load an edge'' that makes the 
 algorithm easy to understand.  Our main theorem, \thmref{thm:main}, compiles this high-level language into a 
 learning graph and bounds the complexity of the resulting quantum query algorithm.  After the theorem is proven, we 
 can design quantum query algorithms using only the high-level language, without reference to learning graphs.  
 This saves the algorithm designer from having to make many repetitive arguments as in \secref{sec:triangle}, and 
 also allows computer search to find the best algorithm within our framework.

\subsection{Special case: subgraph containment}
We now give an overview of our algorithmic framework and its implementation in learning graphs.  We first 
use the framework for computing the function $f_H : [2]^{n \choose 2} \rightarrow \B$, which is by definition $1$ if 
the undirected $n$-vertex input graph contains a copy of some fixed $k$-vertex graph $H=([k], E(H))$ as a subgraph.
This case contains all the essential ideas; after showing this,  it will be easy to generalize 
the theorem in few more steps to any function
$f: [q]^{n \choose 2} \rightarrow \B$ or
 $f: [q]^{n \times n} \rightarrow \B$ with constant-sized $1$-certificate
complexity.

Fix a positive instance $x$, and vertices $a_1, \ldots, a_k \in [n]$ constituting a copy of $H$ in $x$, that is, 
such that $x_{\{a_i, a_j\}}=1$ for all $\{i,j\} \in E(H)$.  Vertices of the learning graph 
will be labeled by $k$-partite graphs with color classes $A_1, \ldots, A_k$. The sets $A_1, \ldots, A_k$ are allowed to 
overlap. Each $L$-vertex label will contain an undirected bipartite graph 
$G_{\{i,j\}} = (A_{\min\{i,j\}}, A_{\max\{i,j\}}, E_{\{i,j\}})$ for every edge $\{i,j\} \in E(H)$,
where $E_{\{i,j\}} \subseteq A_{\min\{i,j\}} \times A_{\max\{i,j\}}$. 
For $\{i,j\} \in E(H)$, 
by $\{a_i, a_j\}$ we mean $(a_i, a_j)$ if
$i <j$, and $(a_j, a_i)$ if $j<i$.
For an edge $\{i,j\} \in E(H)$, and $u \in [n]$, the degree of $u$ in $G_{ij}$ {\em towards} $A_j$ is
the number of vertices in $A_j$ connected to $u$ if $u \in A_i$, and is 0 otherwise.
The edges of these bipartite graphs define naturally the input edges
formally required in the definition of the learning graph:  for $u \neq v$,
both $(u,v)$ and $(v,u)$  define the input edge $\{u,v\}$. We will disregard multiple input edges
as well as self loops corresponding to edges $(u,u)$. Observe that various $L$-vertex labels may correspond to the
same set of input edges. For the ease of notation we will denote $G_{\{i,j\}}$ by both $G_{ij}$ and $G_{ji}$.
We will use similar convention for $E_{\{i,j\}}$ which will be denoted by both $E_{ij}$ and $E_{ji}$. 

Our high-level language consists of three types of commands.  
The first is a setup command.  
This is implemented by choosing sets 
$A_1, \ldots, A_k \subseteq [n]$ of sizes $r_1, \ldots, r_k$ and 
bipartite graphs $G_{ij}$ between $A_i$ and $A_j$ for all $\{i,j\} \in E(H)$.  
Both the set sizes $r_1, \ldots, r_k$ and the 
average degree of vertices in the bipartite graph between $A_i$ and $A_j$ are parameters of 
the algorithm.  The degree parameter $d_{ij}= d_{ji}$ represents the average 
degree of vertices in the smaller of $A_i, A_j$ towards the bigger one in $G_{ij}$. 
It is defined in this fashion so that it is always an integer and at least  one---the average degree of the larger of 
$A_i, A_j$ can be less than one.  
Without loss of generality there is only one setup step and it happens at the beginning of the algorithm.  

The other commands allowed are to load a vertex $a_i$ and to load an edge $\{a_i, a_j\}$ corresponding to 
$\{i,j\} \in E(H)$ (this terminology was introduced by Belovs).  There are two regimes for loading an edge.  One is the dense case, where all vertices in the graph 
$G_{ij}$ have a neighbor; the other is the sparse case, where some vertices in the larger of 
$A_i, A_j$ have no neighbors in the smaller.  We need to separate these two cases as they apparently have 
different costs (and cost analyses).  The algorithm is defined by a choice of set sizes and degree parameters, and 
a loading schedule giving the order in which the vertices and edges are loaded and which loads all edges of~$H$.  

We now define the parameters specifying an algorithm more formally.
\begin{Definition}[Admissible parameters]
Let $H=( [k], E(H))$ be a $k$-vertex graph,
$r_1, \ldots, r_k \in [n]$ be set size parameters, and $d_{ij} \in [n]$ for $\{i,j\} \in E(H)$ be degree parameters.  
Then $\{r_i\}, \{d_{ij}\}$ are {\em admissible} for $H$ if 
\begin{itemize}
\item 
 $1\le r_i \le n/4$ for all $i \in [k]$, 
\item 
 $\ 1 \le d_{ij} \le \max\{r_i, r_j\}$ for all $\{i,j\} \in E(H)$,  
\item 
 for all $i$ there exists $j$ such that $\{i,j\} \in E(H)$ and $d_{ij} (2r_j+1)/(2r_i+1) \ge 1$.  
\end{itemize}
\end{Definition}
We give a brief explanation of the purpose of each of these conditions.  We will encounter terms of the form
${n \choose r_i}/ {n-k \choose r_i}$ that we wish to be $O(1)$; this is ensured by the first condition.  As $d_{ij}$ 
represents the average degree of the vertices in the smaller of $A_i, A_j$ towards the larger, 
the second condition states that this degree cannot 
be larger than the number of distinct possible neighbors.  The third item ensures that the average degree of 
vertices in $A_i$ is at least one in the bipartite graph with some $A_j$.

\begin{Definition}[Loading schedule]
Let $H=( [k], E(H))$ be a $k$-vertex graph with  $m$ edges.  A loading schedule for 
$H$ is a sequence $S=s_1 s_2 \ldots s_{k+m}$ whose 
elements $s_i \in [k]$ or $s_i \in E(H)$ are vertex labels or edge labels of $H$ such that an edge $\{i,j\}$ only appears 
in $S$ after $i$ and $j$, and $S$ contains all edges of $H$.  Let $\Vset_{t}$ be the set of vertices in $S$ before 
position $t$ and similarly $\Eset_{t}$ the set of edges in $S$ before position $t$.
\end{Definition}
We can now state the main theorem of this section.
\begin{theorem}
\label{thm:main}
Let $H=( [k], E(H))$ be a $k$-vertex graph.
Let $r_1, \ldots, r_k, d_{ij}$ be admissible parameters for $H$, and $S$ be a 
loading schedule for $H$.  Then the quantum query complexity of determining if an $n$-vertex graph contains $H$ 
as a subgraph is at most a constant times the maximum of the 
following quantities:
\begin{itemize}
\item Setup cost: $$\sum_{\{u,v\} \in E(H)} \min\{r_u, r_v\} d_{uv},$$
\item Cost of loading $s_t= i$: 
$$\left(\prod_{u \in \Vset_t} \sqrt{\frac{n}{r_u}} \prod_{\{u,v\} \in \Eset_t} \sqrt{\frac{\max\{r_u, r_v\}}{d_{uv}}}\right)
\times  \sqrt{n} \left( \sum_{\substack{j: \{i,j\} \in E(H) \\ r_i \le r_j}} d_{ij} + 
 \sum_{\substack{j : \{i, j\} \in E(H) \\ r_i > r_j}} \frac{r_j d_{ij}}{r_i} \right),
$$
\item Cost of loading $s_t = \{i,j\}$ in the dense case where $(2 \min\{r_i, r_j\}+1) d_{ij} \ge (2 \max\{r_i, r_j\} +1)$: 
$$
\left(\prod_{u \in \Vset_t} \sqrt{\frac{n}{r_u}} \prod_{\{u,v\} \in \Eset_t} \sqrt{\frac{\max\{r_u,r_v\}}{d_{uv}}}\right)
\max\{r_i, r_j\},
$$
\item Cost of loading $s_t = \{i,j\}$ in the sparse case where $(2 \min\{r_i, r_j\}+1) d_{ij} < (2 \max\{r_i, r_j\} +1)$: 
$$
\left(\prod_{u \in \Vset_t} \sqrt{\frac{n}{r_u}} \prod_{\{u,v\} \in \Eset_t} \sqrt{\frac{\max\{r_u,r_v\}}{d_{uv}}}\right) \sqrt{r_i r_j}.
$$
\end{itemize} 
\end{theorem}
If $\{i,j\}$ is loaded in the dense case we call it a type~1 edge, and if it loaded in the sparse case we call it a 
type~2 edge.
The costs of a stage given by \thmref{thm:main} can again be understood more simply in terms of local costs and global 
costs. We give the local and global cost for each stage in the table below.  
\begin{center} 
\begin{tabular}{|r| c c|}
\hline
Stage & Global Cost & Local Cost \\
\hline
Setup & 1 & $\sum_{\{u,v\} \in H} \min\{r_u, r_v\} d_{uv}$ \\
Load vertex $i$ & $\sqrt{n/r_i}$ & $\sqrt{n} \times \text{total degree of } i$ \\
Load a type 1 edge $\{i,j\}$ & $\sqrt{\max\{r_i, r_j\} / d_{ij}}$ & $\max\{r_i , r_j\}$ \\
Load a type 2 edge $\{i,j\}$ & $\sqrt{\max\{r_i, r_j\} / d_{ij}}$ & $\sqrt{r_i  r_j}$ \\
\hline
\end{tabular}
\end{center}

\begin{figure}
\centerline{\includegraphics[scale=0.33]{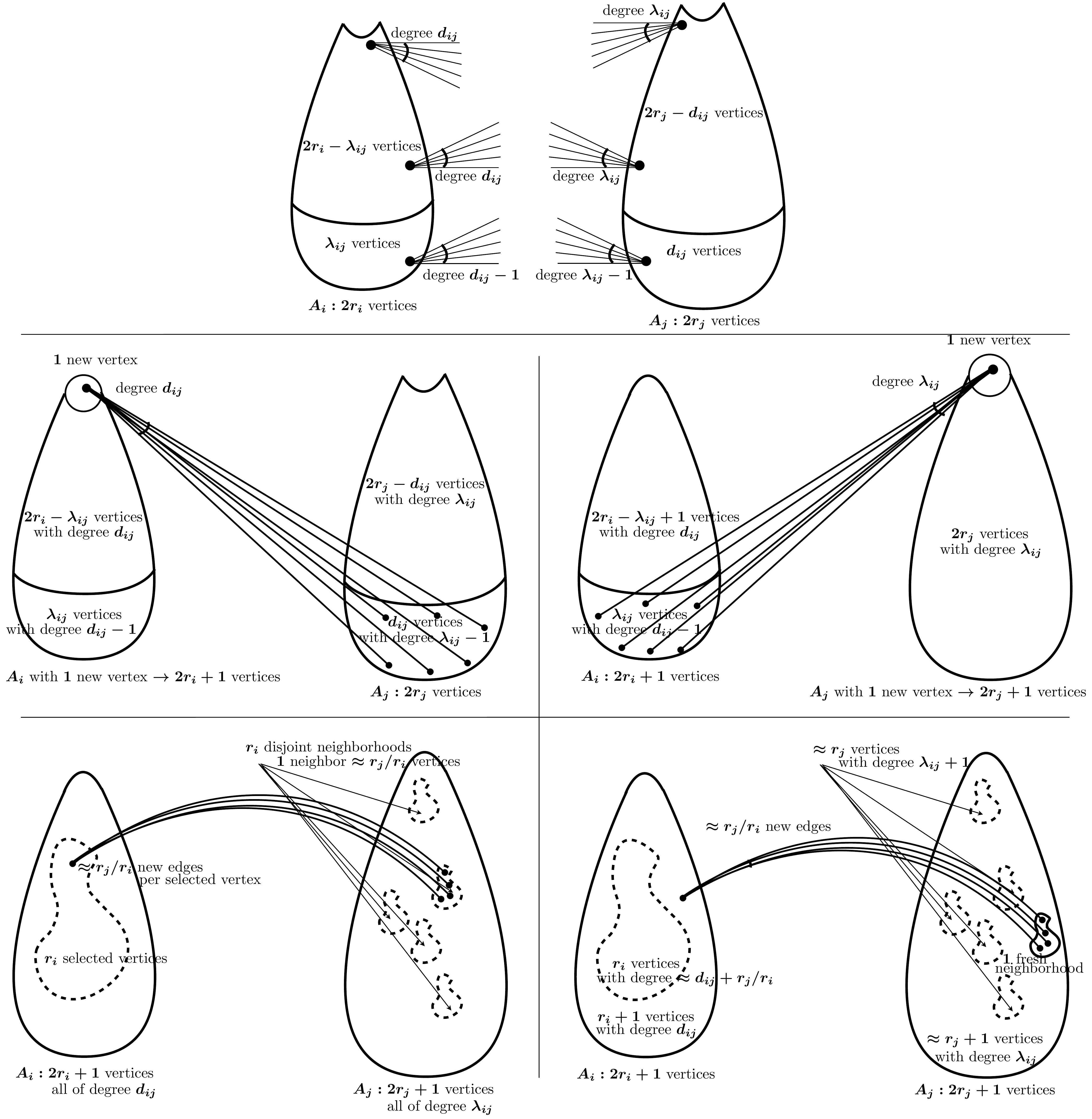}}
\caption{Example of a part of learning graph corresponding to Case~1 and
restricted to the bipartite graph between $A_i$ and $A_j$, where $r_i<r_j$.
Observe that $\lambda_{ij}\approx\frac{r_i}{r_j}d_{ij}$.
The loading schedule is 'setup', 'load $i$', 'load $j$' and 'load $\{i,j\}$'.\label{fig:case1}}
\end{figure}
\begin{figure}
\centerline{\includegraphics[scale=0.33]{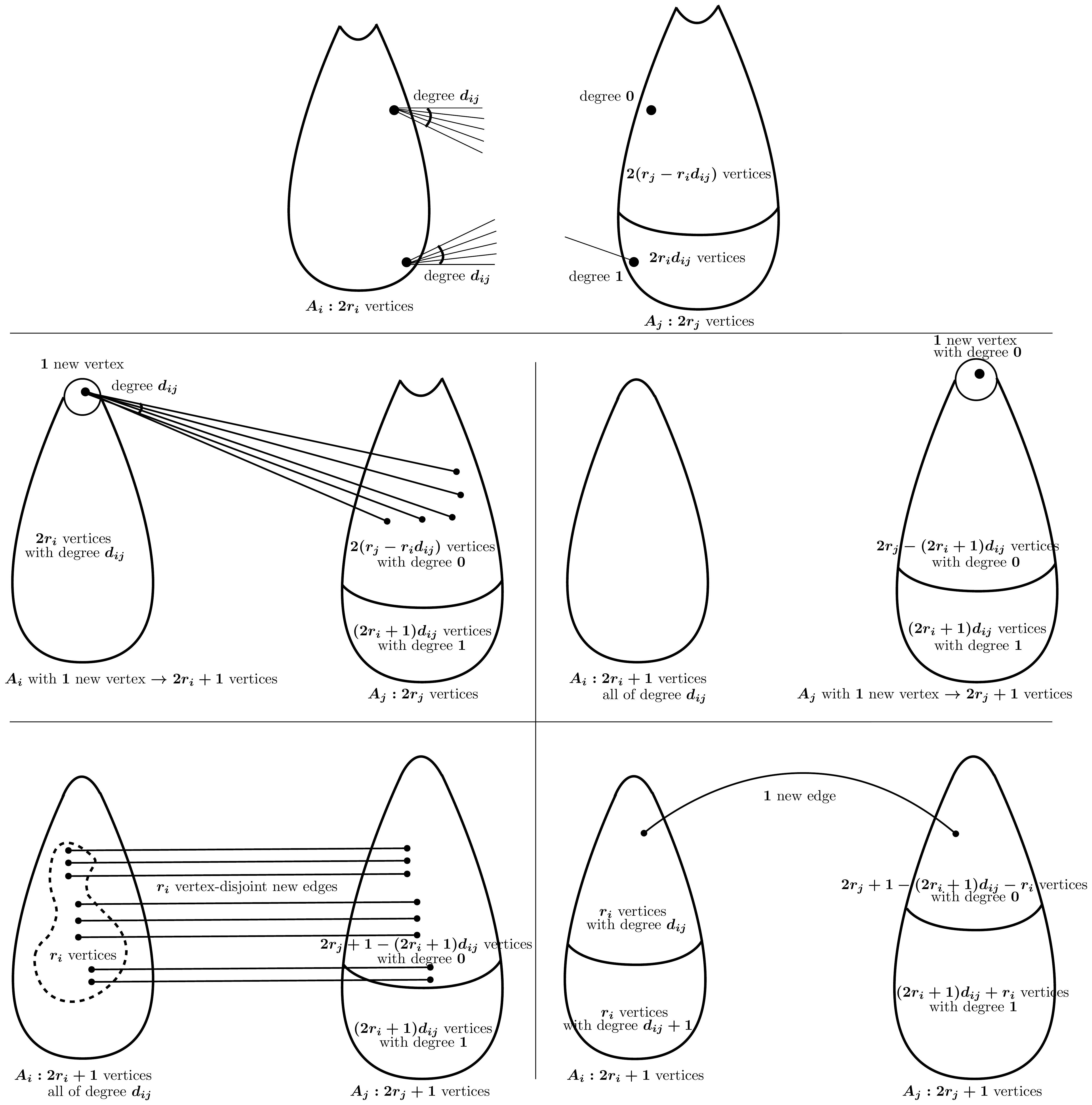}}
\caption{Similar to Figure~\ref{fig:case1}, but for Case~2.\label{fig:case2}}
\end{figure}

\begin{proof}
We show the theorem by giving a learning graph of the stated complexity.  
Vertices of the learning graph 
will be labeled by $k$-partite graphs with color classes $A_1, \ldots, A_k$ of cardinality (of order)
$r_1, \ldots, r_k \in [n]$.  The parameter $d_{ij} \ge 1$ is the average degree of vertices in the smaller of $A_i, A_j$ 
towards the bigger 
in the bipartite graph $G_{ij}$. 

The bipartite graph $G_{ij}$, for each edge $\{i,j\} \in E(H)$, will be 
specified by its type, that is by its degree sequences as given in \defref{def:type}.

We first need to modify the set size parameters $\{r_i\}$ to satisfy a technical condition.  Let 
$r_{i_1} \le \cdots \le r_{i_k}$ be a listing in increasing order.  We set $r_{i_1}' = r_{i_1}$ and 
$r_{i_t}' = \Theta(r_{i_t})$ such that $(2r_{i_t}'+1)/(2r_{i_{t-1}}'+1)$ is an odd integer.  
As a consequence, $(2 \max\{r_i',r_j'\}+1)/(2\min\{r_i', r_j'\}+1)$ is an odd integer, for every $i\neq j$.
We now suppose this is done and drop the primes. 

Throughout the construction of the learning graph we will deal with 
two cases for the bipartite graph between $A_i$ and $A_j$, depending on the size and degree parameters.
\begin{itemize}
\item Case 1 is where $(2 \min\{r_i, r_j\}+1) d_{ij} \ge 2 \max\{r_i, r_j\} +1$, which means that there are enough 
edges from the smaller of $A_i, A_j$ to cover the larger.  We will say that 
the parameters for $\{i,j\}$ are of type 1. In this case, we take $d_{ij}'= \Theta(d_{ij})$ to 
be such that 
\begin{equation}
\label{e:lambda}
2d_{ij}'+1 = (2 \lambda_{ij} + 1) \frac{2 \max\{r_i,r_j\}+1}{2\min\{r_i, r_j\}+1}
\end{equation}
for some integer $\lambda_{ij}$.  This can be done as $(2 \max\{r_i,r_j\}+1)/(2\min\{r_i, r_j\}+1)$ is an odd integer.
In our construction, $\lambda_{ij}$ will be the average degree of the vertices in the larger of 
$A_i, A_j$ towards the smaller, which we want to be integer.  We now consider this done and drop the primes.  
\item Case 2 is where $(2 \min\{r_i, r_j\}+1) d_{ij} < 2 \max\{r_i, r_j\} +1$.  
We will say that the parameters for $\{i,j\}$ are of type 2. In this case, all 
degrees of vertices in the larger of $A_i, A_j$ towards the smaller will be either zero or one.
\end{itemize}

Now we are ready to describe the learning graph.
Figures~\ref{fig:case1} and~\ref{fig:case2} 
illustrate the evolution of a learning graph
for a subsequence $(i,j,\{i,j\})$ of some loading schedule, that is the sequence of instructions
`setup', `load $i$', `load $j$' and `load $\{i,j\}$'. 
The figures only represent the added edges between $A_i$ and $A_j$, where $r_i<r_j$.
Figure~\ref{fig:case1} corresponds to Case~1, and 
Figure~\ref{fig:case2} to Case~2.

Recall that for every positive instance $x$, 
we fixed $a_1, \ldots, a_k \in [n]$ be such that $x_{\{a_u, a_v\}}=1$ for all $\{u,v\} \in E(H)$.
During the construction we will specify   for every edge $\{u,v\} \in E(H)$, and for every stage number $t$,
the {\em correct degree}
$\cd(u,v,t)$ which is the degree of $a_u$ in  $G_{ij}$ towards $a_v$ in each  $L$-vertex of $V_{t+1}$ with positive flow. 

\paragraph{Stage 0 (Setup):}
For each edge $\{i,j\} \in E(H)$ we setup a bipartite graph between $A_i$ and $A_j$.  The type of the 
bipartite graph depends on the type of the parameters for $\{i,j\}$.  Let $\ell = \min\{r_i, r_j\}$ and $g = \max\{r_i, r_j\}$.
\begin{itemize}
\item Case 1:   Solving for $\lambda_{ij}$ in \eqnref{e:lambda} we get $\lambda_{ij} = 
((2\ell + 1) d_{ij} + \ell- g) / (2g + 1)$.  Intuitively, $d_{ij}$ represents the average degree of vertices in the smaller 
of $A_i, A_j$ and $\lambda_{ij}$ the average degree in the larger.  
Formally, the type of bipartite graph between $A_i,A_j$, with the listing of degrees for the smaller set given first, is
$(\{(2\ell-\lambda_{ij}, d_{ij}), (\lambda_{ij}, d_{ij}-1)\}, 
\{(2g-d_{ij}, \lambda_{ij}),(d_{ij}, \lambda_{ij}-1)\})$.

\item Case 2: In this case the type of bipartite graph between $A_i$ and $A_j$, with the listing of degrees for the smaller 
set given first, is $(\{(2\ell, d_{ij})\}, \{(2\ell d_{ij}, 1 ), (2g - 2\ell d_{ij},0)\})$.
\end{itemize}
The $L$-vertices at the end of stage $0$ will be labeled by (possibly overlapping) sets 
$A_1, \ldots, A_k$ of sizes $r_1, \ldots, r_k$ and edges corresponding to a graph of the appropriate type between 
$A_i$ and $A_j$ for all $\{i,j\} \in E(H)$.  Flow goes uniformly to those $L$-vertices where none of 
$a_1, \ldots, a_k$ are in any of the sets $A_1, \ldots, A_k$.  For all $\{u,v\} \in E(H)$, we set 
$\cd(u,v,0) = 0$.

\paragraph{Stage $t$ when $s_t = i$:} In this stage we load $a_i$.  The $L$-edges in this stage select a vertex 
$v$ and add it to $A_i$.  For all $j$ such that $\{i,j\} \in E(H)$ we add the following edges:
\begin{itemize}
\item Case 1: Say the parameters for $\{i,j\}$ are of type 1.
If $r_i \le r_j$, then $v$ is connected 
to those vertices of degree $\lambda_{ij}-1$ in $A_j$, and we set $\cd(i,j,t) = d_{ij}$.
Otherwise 
$v$ is connected to those vertices of degree $d_{ij} -1$ in $A_j$, and we set $\cd(i,j,t) = \lambda_{ij}$. 
\item Case 2:  Say the parameters for $\{i,j\}$ are of type 2.
If $r_i \le r_j$ then $v$ is connected to $d_{ij}$ vertices of degree 0 in $A_j$, and we set $\cd(i,j,t) = d_{ij}$.  
Else 
no edges are added between $v$ and $A_j$, and we set $\cd(i,j,t) = 0$.
\end{itemize}

For all other $(u,v)$, we set $\cd(u,v,t) = \cd(u,v,t-1)$.
Flow goes uniformly on those $L$-edges where $v=a_i$.  

\paragraph{Stage $t$ when $s_t = \{i,j\}$:} In this stage we load $\{a_i, a_j\}$.  
Again we break down according to the type of the parameters for $\{i, j\}$.
Let $\ell = \min\{r_i, r_j\}$ and 
$g = \max\{r_i, r_j\}$.

\begin{itemize}
\item Case 1: As both $a_i$ and $a_j$ have been 
loaded, between $A_i$ and $A_j$ there is a bipartite graph of type 
$(\{(2\ell+1, d_{ij})\}, \{(2g+1, \lambda_{ij})\})$, with the degree listing of the smaller set 
coming first.  If we simply added $\{a_i, a_j\}$ at this step, $a_i$ and $a_j$ would be uniquely identifiable by their 
degree and blow up the complexity of later stages.

To combat this, loading $\{a_i, a_j\}$ will consist of two substages $t.I$ and $t.II$.  The first substage is a 
hiding step, done to reduce the complexity of having $\{a_i, a_j\}$ loaded.  Then we actually load $\{a_i, a_j\}$. 

Substage $t.I$: Let $h = (2g+1)/(2\ell+1)$.  We select $\ell$ vertices in the smaller of $A_i, A_j$, and to each of 
these add $h$ many neighbors.  All neighbors chosen in this stage are distinct.  Thus at the end of this stage the type 
of bipartite graph between $A_i$ and $A_j$ is $(\{(\ell, d_{ij}+h), (\ell+1, d_{ij})\}, 
\{(\ell (2g+1)/(2\ell+1), \lambda_{ij}+1), ((2g+1)(1-\ell/(2\ell+1)), \lambda_{ij})\})$.  Flow goes uniformly along those 
$L$-edges where neither $a_i$ nor $a_j$ receive any new edges.  For all $\{u,v\} \in E(H)$, we set $\cd(u,v,t+1) = \cd(u,v,t)$.

Substage $t.II$:  The $L$-edges in this substage select a vertex 
$u$ in the smaller of $A_i, A_j$ of degree $d_{ij}$ and add $h$ many neighbors of degree $\lambda_{ij}$.  
Flow goes uniformly along those $L$-edges where $u \in \{a_i, a_j\}$ and $\{a_i, a_j\}$ is one of the edges 
added. Let $s$ be the index of the smaller of the sets $A_i, A_j$, and let $b$ the other index.
We set $\cd(s,b,t+1) = d_{ij} +h, \cd(b,s,t+1) = \lambda_{ij} +1$ and $\cd(u,v,t+1) = \cd(u,v,t)$ for
$\{u,v\} \neq \{i,j\}$.

\item Case 2:  As both $a_i$ and $a_j$ have been loaded, there is a bipartite graph of type 
$(\{(2\ell+1, d_{ij})\}, \{((2\ell+1)d_{ij},1), (2g+1 - (2\ell+1)d_{ij},0)\})$.  We again first do a hiding step, and then add
the edge $\{a_i, a_j\}$.

Substage $t.I$: We select $\ell$ vertices in the smaller of $A_i, A_j$ and to each add a single edge to a vertex 
of degree zero in the larger of $A_i, A_j$.  Flow goes uniformly along those $L$-edges where no edges adjacent 
to $a_i, a_j$ are added. For all $\{u,v\} \in E(H)$, we set $\cd(u,v,t+1) = \cd(u,v,t)$.

Substage $t.II$: A single edge is added between a vertex in the smaller of $A_i, A_j$ of degree $d_{ij}$ and a 
vertex in the larger of $A_i, A_j$ of degree zero.  Flow goes along those $L$-edges where $\{a_i, a_j\}$ is 
added.  Let again $s$ be the index of the smaller of the sets $A_i, A_j$, and let $b$ the other index.
We set $\cd(s,b,t+1) = d_{ij} +1, \cd(b,s,t+1) = 1$ and $\cd(u,v,t+1) = \cd(u,v,t)$ for
$\{u,v\} \neq \{i,j\}$.
\end{itemize}

This completes the description of the learning graph.

\paragraph{Complexity analysis}
We will use \lemref{e:simple_costg} to evaluate the complexity of each stage.  First we need to establish the 
hypothesis of this lemma, which we will do using \lemref{lem:trans_cons} and \lemref{e:symmetry}.
Remember that given $\sigma \in S_n$, we defined and denoted by $\sigma$ the permutation over $[n]\times[n]$ 
such that $\sigma(i,j)=(\sigma(i),\sigma(j))$.
First of all let us observe that every $\sigma \in S_n$ is in the automorphism group of the function we are computing,
since it maps a 1-certificate into a 1-certificate.  As the flow only depends on the 1-certificate graph, this implies 
that $S_n$ acts transitively on the flows and therefore we obtain the conclusion of \lemref{lem:trans_cons}.

Let $V_t$ stand for the $L$-vertices at the beginning of stage $t$. 
For a positive input $x$, and for an $L$-vertex $P \in V_t$,
we will denote the incoming flow to $P$ on $x$ by $p_x(P)$ and the number of outgoing edges from $P$ with
positive flow on $x$ by $g_x^+(P)$. For an $L$-vertex $R \in V_{t-1}$ 
we will denote by $g_{x,R}^-(P)$ number of incoming edges to $P$ from $L$-vertices of  the isomorphism type of $R$ with
positive flow on $x$, that is  $g_{x,R}^-(P) = |\{ \tau \in S_n : p_x(\tau(R),P) \neq 0 \}|$.
The crucial features of our learning graph construction are
the following:  at every stage, for every $L$-vertex $P$ and every $\sigma \in S_n$, the $L$-vertex $\sigma(P)$ is 
also present.
The outgoing flow from an $L$-vertex is always uniformly distributed among the edges getting  flow.
The flow depends only on the vertices in the input containing a copy of the graph $H$, and therefore the values
$g_x^+(P)$ and $g_{x,R}^-(P)$, for $p_x(P)$ non-zero, 
depend only on the isomorphism types of $P$ and $R$. Mathematically,
this last property translates to:
for all $t$, for all $P \in V_t$, for all $R \in V_{t-1}$, for all positive inputs $x$ and $y$, for all $\sigma \in S_n$, we have
\begin{equation}
\label{e:basic}
[ p_x(P) \neq 0 \mbox{{\rm ~and~}} p_y(\sigma(P)) \neq 0 ] 
\ \ \ \Longrightarrow \ \ \ 
[ g_x^+(P) = g_y^+(\sigma(P)) \mbox{{\rm ~and~}} g_{x,R}^-(P) = g_{y,R}^-(\sigma(P)) ].
\end{equation}
which is exactly the hypothesis of \lemref{e:symmetry}.

Now we have established the hypotheses of \lemref{e:simple_costg} and turn to evaluating the bound given there.  
The main task is evaluating the maximum vertex ratio of each stage.  The general way we will do this is to 
consider an arbitrary vertex $P$ of a stage.  We then lower bound the probability that $\sigma(P)$ is in the flow 
for a positive input $x$ and a random permutation $\sigma\in S_n$, 
without using any particulars of $P$.  This will then upper bound the maximum vertex 
ratio.  We use the notation $P \in F_x$ to denote that $L$-vertex $P$ has at least one incoming edge with flow on 
input $x$.  

\begin{lemma}[Maximum vertex ratio]\label{lem:vertexflow}
For any $L$-vertex $P \in V_{t+1}$ and any positive input $x$
\[
\Pr_\sigma[\sigma(P) \in F_x] =
\Omega \left(\prod_{j \in \Vset_t} \frac{r_j}{n} \prod_{(u,v) \in \Eset_t} \frac{d_{uv}}{\max\{r_u, r_v\}}\right).
\]
\end{lemma}
\begin{proof}
We claim that an $L$-vertex $P$ in $V_{t+1}$, that is at the end of stage $t$, has flow if and only if
\begin{eqnarray}
\label{e:present}
&&\forall  i \in \Vset_t, ~ \forall  \{i,j\} \in \Eset_t, 
  \text{~we~have~} a_i \in A_i  \mbox{{\rm ~and~}} \{a_i, a_j\} \in E_{ij},\\
&&\label{e:absent}
 \forall  i \in [k] \setminus \Vset_t, ~ \forall  \{i,j\} \in H(E) \setminus \Eset_t,
 \text{~we~have~} a_i \not \in A_i  \text{~and~} \{a_i, a_j\} \not \in E_{ij},\\
&&\label{e:correct}
\forall \{i,j\} \in \Eset_t,
\text{~the~degree~of~} a_i  \mbox{{\rm ~in ~}} G_{ij} \text{~towards ~} A_j \text{~is ~} \cd(i,j,t).
\end{eqnarray}
The only if part of the claim is obvious by the construction of the learning graph. The if part can be proven
by induction on $t$. For $t =0 ,$ 
the first half (\ref{e:absent}) is exactly the one
which defines the flow for $L$-vertices in $V_1$. 

For the inductive step let us suppose first that $s_t = i$.
Consider the label $P'$ by dropping the vertex $a_i$ from $A_i$.  
Then in $P'$ every bipartite graph is of appropriate type for level $t$ because of (\ref{e:correct}), and
therefore $P' \in V_t$. It is easy to check that $P'$ also satisfies all three conditions,
(for (\ref{e:correct}) we also have to use the second half of (\ref{e:absent}): $\{a_i, a_j\} \not \in E_{ij}$), and therefore has positive flow.
Since $P'$ is a predecessor of $P$ is the learning graph, $P$ has also positive flow.

Now let us suppose that $s_t = \{i,j\}$. In $P$ the edge set $E_{ij}$ can be decomposed into the disjoint union
of $E_1 \cup E_2$, where $E_1$ a bipartite graph of type $(\{(2 \ell+1, d_{ij})\}, 
\{(2g+1, \lambda_{ij})\})$ and $E_2$ is of type
$(\{( \ell+1, h), (\ell , 0\}, 
\{((\ell +1)h, 1), (2g+1 - ( \ell +1)h, 0)\})$, and (\ref{e:correct}) implies that $\{a_i, a_j\} \in E_2$.
Consider the label $P'$ by dropping the edges of $E_2$ from $E_{ij}$. Again, $P'$ satisfies the inductive
hypotheses, and therefore gets positive flow, which implies the same for $P$.

Suppose now that the $L$-vertex $P$ is labeled by sets $A_1, \ldots, A_k$ (some may be empty) and let the 
set of edges between $A_i$ and $A_j$ be $E_{ij}$.  We want to lower bound the probability that 
$\sigma(P) \in F_x$, meaning that $\sigma(P)$ satisfies the above three conditions.  Item~(\ref{e:absent}) is 
always satisfied with constant probability; moreover, conditioned on item~(\ref{e:absent}) the probability of the 
other events does not decrease.  Thus we take this constant factor loss and focus on the 
items~(\ref{e:present}), (\ref{e:correct}). 

We also claim that, conditioned on item~(\ref{e:present}) holding, item~(\ref{e:correct}) 
holds with constant probability.  This can be seen 
as in the hiding step, in both case~1 and case~2, the probability that $a_i, a_j$ have the correct degree given that 
they are loaded is at least $1/4$.  In the step of loading an edge, again in case~1 half the vertices on the left and 
right hand sides have the correct degree and so this probability is again $1/4$; in case~2, given that the edge is 
loaded, whichever of $a_i,a_j$ is in the larger set will automatically have the correct degree, and the other one will 
have correct degree with probability $1/2$.  Now we take this constant factor loss to obtain that  
$\Pr_\sigma[\sigma(P) \in F_x]$ is lower bounded by a constant factor times the probability that 
item~(\ref{e:present}) holds.

The events in the first condition are independent, except that for the edge $\{a_i, a_j\}$ to be loaded
the vertices $a_i$ and $a_j$ have to be also loaded.
Thus we can lower bound the probability it is 
satisfied by
$$
\Pr_\sigma[\sigma(P) \in F_x] = \Omega\Big(\prod_{i \in VS_t} \Pr_\sigma[a_i \in \sigma(A_i)] 
\times\!\!\! \prod_{(u,v) \in ES_t}  \Pr_\sigma[\{a_i, a_j\} \in \sigma(E_{ij}) | a_i \in \sigma(A_i), a_j \in \sigma(A_j)] \Big)
$$
Now $\Pr_\sigma[a_i \in \sigma(A_i)] = \Omega(r_i/n)$ as this fraction of permutations will put 
$a_i$ into a set of size $r_i$.  
For the edges we use the following lemma.
\begin{lemma}
\label{l:uniform}
Let $Y_1,Y_2 \subseteq [n]$ be of size $\ell, g$ respectively, and 
let $(y_1,y_2)\in Y_1\times Y_2$.
Let $K$ be a bipartite graph between $Y_1$ and $Y_2$
of type $\{(\ell, d)\}, \{(g, \ell d /g)\}$.
Then $\Pr_\sigma[\{y_1,y_2\} \in \sigma(K) =d/ g$.
\end{lemma}
\begin{proof}
Because of symmetry, this probability does not depend on the choice of $\{y_1, y_2\}$;
denote it by $p$. Let $K_1, \ldots, K_c$ be an enumeration of all bipartite graphs isomorphic to $K$. 
We will count in two different ways the cardinality $\chi$ of the set
$\{ (e,h) : e \in K_h \}$. Every $K_h$ contains $\ell d$ edges, therefore $\chi = c\ell d$.  
On the other hand, every edge appears in $pc$ graphs, therefore $\chi = \ell g pc$, and thus $p = d/g$.
\end{proof}
In our case, the graph $G_{ij}$ as in the hypothesis of the lemma plus some additional edges.  By 
monotonicity, it follows that 
$$
\Pr_\sigma[\{a_i, a_j\} \in \sigma(E_{ij}) | a_i \in \sigma(A_i), a_j \in \sigma(A_j)]=\Omega(d_{ij}/ \max\{r_i, r_j\}).
$$
\end{proof}

This analysis is common to all the stages.  Now we go through each type of stage in turn to evaluate the 
stage specific length and degree ratio.  
\paragraph{Setup Cost:} 
The length of this stage is upper bounded by 
$$\sum_{(i,j) \in E(H)} \min\{r_i, r_j\} d_{ij}.$$  We can upper bound the degree ratio by 
$$\prod_{i \in [k]} \frac{{n \choose 2r_i}}{{n-k \choose 2r_i}} \le 2^{k} = O(1)$$
as $r_i < n/4$.

\paragraph{Stage $t$ when $s_t=i$:} 
In a stage loading a vertex the degree ratio is $O(n)$ as there are $n-r_i$ possible vertices to 
add yet only one is used 
by the flow.  The length of this stage is the total degree which is upper bounded by
$$ \sum_{\substack{j: \{i,j\} \in E(H) \\ r_i \le r_j}} d_{ij} +
 \sum_{\substack{j : \{i, j\} \in E(H) \\ r_i > r_j}} \frac{r_j d_{ij}}{r_i}.
$$

\paragraph{Stage $t$ when $s_t=\{i,j\}$:} Technically we should analyze the complexity of the two substages as two distinct
stages. However, as we will see, in both cases the degree ratio in the first substage is $O(1)$,
and therefore the local cost of this stage is just the maximum of the local cost of the two substages.

\paragraph{Stage $t.I$
:}
In Case~1, the length of this stage is $O(\max\{r_i, r_j\})$ and the degree ratio is constant.
In Case~2, the length of this stage is $O(\min\{r_i, r_j\})$ and the degree ratio is constant.

\paragraph{Stage $t.II$
:}
In Case~1, the length is $h=O(\max\{r_i, r_j\}/\min\{r_i, r_j\})$.  The degree ratio is of order
$\ell \frac{{g \choose h}}{{g-1 \choose h-1}} =O(\ell^2)$.
Thus the square root of the degree ratio times the length is of order $\max\{r_i, r_j\}$.

In Case~2, the length is one and the degree ratio is $O(r_i r_j)$ as there are $O(r_i r_j)$ 
many possible edges that could be added and the flow uses one.  

Thus in Case~1 in both substages the product of the length and square root of degree ratio is 
$O(\max\{r_i, r_j\})$.  In Case~2, substage II dominates the complexity where the product of the length and square root 
of degree ratio is $O(\sqrt{r_i r_j})$.
\end{proof}

\subsection{Extensions and basic properties}\label{sec:extension}
We now  extend \thmref{thm:main} to the general case of computing a function 
$f: [q]^{n \times n} \rightarrow \B$ with constant-sized $1$-certificates.  A certificate graph for such a function 
will be a directed graph possibly with self-loops.  Between $i$ and $j$ there can be bidirectional edges, that is 
both $(i,j)$ and $(j,i)$ present in the certificate graph, but there will not be multiple edges between $i$ and $j$, as 
there are no repetitions of indices in a certificate.  

We start off by modifying the algorithm of \thmref{thm:main} to work for detecting directed graphs with possible 
self-loops.  To do this, the following transformation will be useful.
\begin{Definition}
Let $H$ be a directed graph, possibly with self-loops.  The undirected version $U(H)$ of $H$ 
is a simple undirected graph formed by eliminating any self-loops in $H$, and making all edges of $H$ 
undirected and single.
\end{Definition}   

\begin{lemma}
\label{lem:directed}
Let $H$ be a directed $k$-vertex graph, possibly with self loops.  Then the quantum 
query complexity of detecting if an $n$-vertex directed graph $G$ contains $H$ as a subgraph is at most a 
constant times the complexity given in~\thmref{thm:main} of detecting $U(H)$ in an $n$-vertex undirected graph.  
\end{lemma}
\begin{proof} 
Let $H$ be a directed $k$-vertex graph (possibly with self-loops) and $H'=U(H)$ be its undirected version.
Let $r_1, \ldots, r_k, d_{ij}$ be admissible parameters for $H'$, and $S$ a loading schedule for $H'$.  Fix a directed 
$n$-vertex graph $G$ containing $H$ as a subgraph.  Let $a_1, \ldots, a_k$ be vertices of $G$ such that
$(a_i, a_j) \in E(G)$ for $(i,j) \in E(H)$.  
We convert the algorithm for loading $H'$ in \thmref{thm:main} into one for loading $H$ of the same complexity.  

The setup step for $H'$ is modified as follows.  In the bipartite graph between $A_i$ and $A_j$, if both 
$(i,j), (j,i) \in E(H)$ then all edges between $A_i$ and $A_j$ are directed in both directions; otherwise, 
if $(i,j) \in E(H)$ or $(j,i) \in E(H)$ they are 
directed from $A_i$ to $A_j$ or vice versa, respectively.  For every self-loop 
in $H$, say $(i,i) \in E(H)$, we add self-loops to the vertices in $A_i$.  Note that these modifications at most 
double the number of edges added, and hence the cost, of the setup step.

Loading a vertex: When loading $a_i$ we connect it as before, now orienting the 
edges according to $(i,j)$ or $(j,i)$ in $E(H)$, or both.  If $(i,i) \in E(H)$, then we add a self loop to $a_i$.  
The only change in the complexity of this stage is again the length, which at most doubles.  Notice that in the case 
of a self-loop we have also already loaded the edge $(a_i, a_i)$.  We do not incur an extra cost for loading 
this edge, however, as the self loop is loaded if and only if the vertex is.

Loading an edge: Say that we are at the stage where $s_t = \{i,j\} \in E(H')$.  If exactly one of $(i,j), (j,i) \in E(H)$ 
then this step happens exactly as before, except that the bipartite graph has edges directed from $A_i$ to $A_j$ or 
vice versa, respectively.  If both $(i,j)$ and $(j,i) \in E(H)$, then in this step all edges added are bidirectional.  
This again at most doubles the length, and does not affect the degree flow probability as $(a_i, a_j)$ is loaded 
if and only if $(a_j, a_i)$ is loaded as all edges are bidirectional.
\end{proof}

\begin{lemma}
\label{lem:f}
Let $f: [q]^{n \times n} \rightarrow \{0,1\}$ be a function such that all minimal $1$-certificate graphs are isomorphic to 
a directed $k$-vertex graph $H$.  
Then the quantum query complexity of computing $f$ is at most the complexity of detecting $H$ in an $n$-vertex 
graph, as given by \lemref{lem:directed}. 
\end{lemma}

\begin{proof} 
We will show the theorem by giving a learning graph algorithm.  Let $\G=(\V,\E,S,w,\{p_y\})$ be the learning graph 
from \lemref{lem:directed} for $H$.  All of $\V,\E,S,w$ will remain the same in our learning graph $\G'$ for $f$.  
We now describe the definition of the flows in $\G'$.

Consider a positive input $x$ to $f$, and let $\alpha$ be a minimal $1$-certificate for $x$ such that the certificate graph $H_\alpha$ is 
isomorphic to $H$.  The flow $p_x$ will be defined as the flow for $H_\alpha$ (thought of as an 
$n$-vertex graph, thus with $n-k$ isolated vertices) in $\G$, the learning graph for detecting $H$.  
This latter flow has the property that the label of every terminal of flow contains $E(H_\alpha)$ and thus 
will also contain the index set of a $1$-certificate for $x$.  

The positive complexity of the learning graph for $f$ will be the same as that for detecting $H$ and the negative 
complexity will be at most that as in the learning graph for detecting $H$, thus we conclude that the complexity 
of computing $f$ is at most that for detecting $H$ as given in \lemref{lem:directed}.
\end{proof}

\begin{theorem}
\label{thm:several}
Say that the $1$-certificate complexity of $f: [q]^{n \times n} \rightarrow \{0,1\}$ is at most a constant $m$, and let 
$H_1, \ldots, H_c$ be the set of graphs (on at most $m$ edges) for which there is some positive input $x$ 
such that $H_i$ is a minimal $1$-certificate graph for $x$.  Then the quantum query complexity of computing $f$ is at most a constant times the maximum of 
the complexities of detecting $H_i$ for $i=1,\ldots, c$ as given by \lemref{lem:directed}.
\end{theorem}
\begin{proof}
Consider learning graphs $\mathcal{G}_1, \ldots, \mathcal{G}_c$ given by \lemref{lem:directed} for detecting 
$H_1,\ldots, H_c$ respectively.  Further suppose these learning graphs are normalized such that their 
negative and positive complexities are equal.  

We construct a learning graph $\mathcal{G}$ for $f$ where the edges and vertices are given by connecting a 
new root node by an edge of 
weight one to the root nodes of each of $\mathcal{G}_1, \ldots, \mathcal{G}_c$.  Thus the negative 
complexity of $\mathcal{G}$ is at most $c(1 + \max_i C_0(\mathcal{G}_i)$.  

Now we construct the flow for a positive input $x$.  Let $\alpha$ be a minimal $1$-certificate for $x$ such that the 
certificate graph $H_\alpha$ is isomorphic to $H_i$, for some $i$.  Then the flow on $x$ is first directed 
entirely to the root node of $\mathcal{G}_i$.  It is then defined within $\mathcal{G}_i$ as in \lemref{lem:f}.
Thus the positive complexity of $\mathcal{G}$ is at most $c(1 + \max_i C_1(\mathcal{G}_i))$.  
\end{proof}

To make \thmref{thm:main} and \lemref{lem:directed} easier to apply, here we establish some basic intuitive properties about the 
complexity of the algorithm for different subgraphs.  Namely, we show that if $H' $ is a subgraph of $H$ then 
the complexity given by \lemref{lem:directed} for detecting $H'$ is at most that of $H$.  We show a similar statement 
when $H'$ is a vertex contraction of $H$.

\begin{lemma}\label{lem:sub}
Let $H$ be a directed $k$-vertex graph (possibly with self-loops) and $H'$ a subgraph of $H$.  Then the 
quantum query complexity of determining if an $n$-vertex graph $G$ contains $H'$ is at most that of determining 
if $G$ contains $H$ from \lemref{lem:directed}.
\end{lemma}

\begin{proof} 
Assume that the vertices of $H$ are labeled from $[k]$ and that $H'$ is labeled such $(i,j) \in E(H)$ for all 
$(i,j) \in E(H')$.

The learning graph we use for detecting $H'$ is the same as that for $H$.  For a graph $G$ containing a $H'$ as 
a subgraph, let $a_1, \ldots, a_k$ be such that $(a_i, a_j) \in G$ for all $(i, j) \in H'$.  (If $t$ is an isolated vertex in 
$H'$, then $a_t$ can be chosen arbitrarily).  The flow for $G$ is defined 
in the same way as in the learning graph for $H$.  Note that once 
$a_1, \ldots, a_k$ have been identified, the definition of flow depends only edge slots---not on edges---thus 
this definition remains valid for $H'$.  Furthermore all terminals of flow are labeled by edge slots
$(a_i, a_j)$ for all $(i, j) \in H$, and so also contain the edge slots for $H'$.  Thus this is a valid flow for 
detecting $H'$.  As the learning graph and flow are the same, the complexity will be as that given in 
\lemref{lem:directed}.
\end{proof}

\begin{lemma}
\label{lem:contract}
Let $H$ be a $k$-vertex graph and $H'$ a vertex contraction of $H$.  Then the quantum query complexity of 
detecting $H'$ is at most that of detecting $H$ given in \lemref{lem:directed}.
\end{lemma}

\begin{proof} 
Again we assume that the vertices of $H$ are labeled from $[k]$.  The key point is the following: if 
$H'$ is a vertex contraction of $H$, then there are $z_1, \ldots, z_k \in [k]$ (not necessarily distinct) such that 
$(z_i, z_j) \in E(H')$ if and only if $(i,j) \in E(H)$.  The learning graph for $H'$ will be the same as that for $H$ except for the flows.
For a graph $G$ containing $H'$, we choose vertices $a_1, \ldots, a_k$ (not necessarily distinct) such that 
if $(z_i, z_j) \in E(H')$ then $(a_i, a_j) \in E(G)$.  As $(z_i, z_j) \in E(H')$ if and only if $(i,j) \in E(H)$, we can define 
the flow as in \lemref{lem:directed} for $a_1, \ldots, a_k$ to load a copy of $H'$.  (Note that there is no restriction in 
the proof of that theorem that the sets $A_1, \ldots, A_k$ be distinct).  This gives an algorithm for detecting 
$H'$ with complexity at most that given by \lemref{lem:directed} for detecting $H$.
\end{proof}

\section{Associativity testing}
Consider an operation $\circ : S \times S \rightarrow S$ and let $n=|S|$.  We wish to determine if $\circ$ is 
\emph{associative} on $S$, meaning that $a \circ (b \circ c) = (a \circ b) \circ c$ for all $a,b,c \in S$.  We are given 
black box access to $\circ$, that is, we can make queries of the form $(a,b)$ and receive the answer 
$a \circ b$.  

\begin{theorem} 
Let $S$ be a set of size $n$ and $\circ : S \times S \rightarrow S$ be an operation that can be accessed in black-box 
fashion by queries $(a,b)$ returning $a \circ b$.  There is a bounded-error quantum query 
algorithm to determine if $(\circ, S)$ is associative making $O(n^{10/7})$ queries.    
\end{theorem}

\begin{proof}
If $\circ$ is not associative, then there is a triple $a_2,a_3,a_4$ such that 
$a_2 \circ (a_3 \circ a_4) \ne (a_2 \circ a_3) \circ a_4$.  
A certificate to the non-associativity of $\circ$ is given by
$a_3\circ a_4 = a_1, a_2 \circ a_1, a_2 \circ a_3 =a_5$, and $a_5 \circ a_4$ such that 
$a_2 \circ a_1 \ne a_5 \circ a_4$ (see Figure~\ref{fig:assoc}).  Note that not all of 
$a_1, \ldots, a_5$ need to be distinct.  

Let $H$ be a directed graph on $5$ vertices with directed edges $(2,1), (2,3), (3,4), (5,4)$.  Each non-associative 
input has a certificate graph either isomorphic to $H$ or a vertex contraction of $H$, in the case that not all of
$a_1, \ldots, a_5$ are distinct.  By \lemref{lem:contract}, the complexity of a detecting a vertex contraction of $H$ 
is dominated by that of detecting $H$, and so by \thmref{thm:several} it suffices to show the theorem for $H$.

\begin{figure}
\centerline{\includegraphics[scale=0.4]{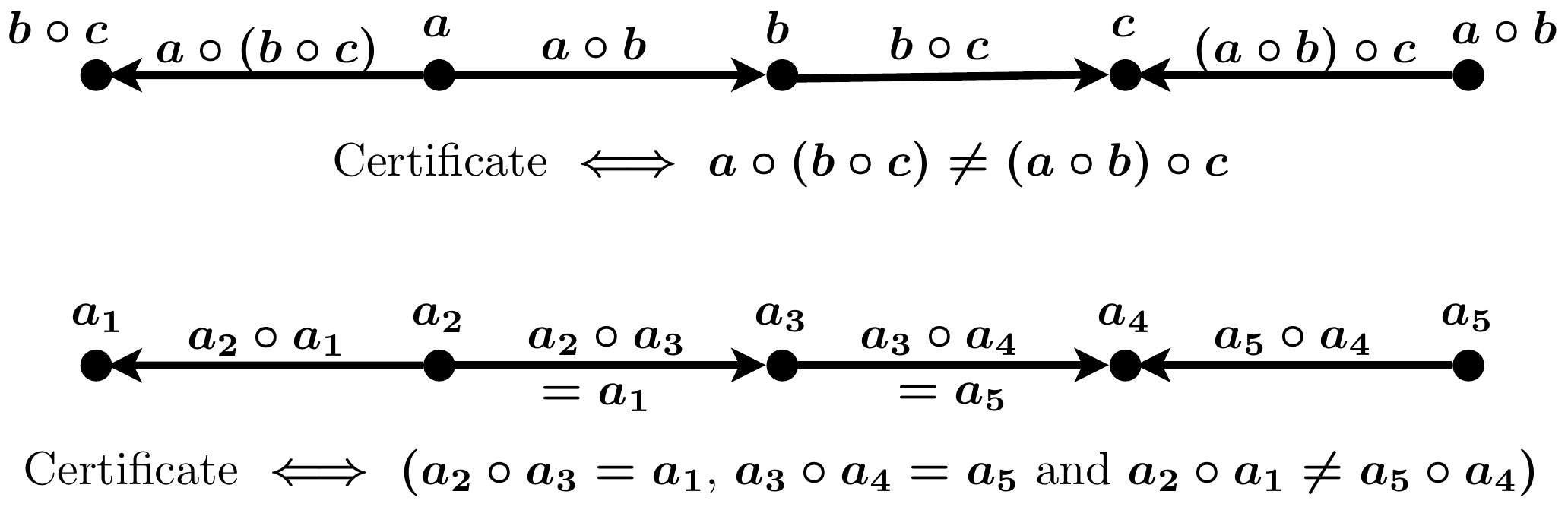}}
\caption{The $5$-vertex certificate graph for associativity. Both pictures represent the same graph certificate $H$, where
the second one has been labelled according to the notations of our abstract language.\label{fig:assoc}} 
\end{figure}

We use the algorithmic framework of \thmref{thm:main} to load the graph $H$.  Let 
$r_1 = n/10, r_2 = n^{4/7}, r_3=n^{6/7}, r_4=n^{5/7}, r_5=1$ 
and $d_{21}=n^{6/7}, d_{23}=n^{6/7}, d_{34}=n^{5/7}, d_{54}=1$.  Here $d_{ij}$ indicates the average degree of 
vertices in the smaller of $A_i, A_j$ for edges directed from $A_i$ to $A_j$.  
It can be checked that this is an admissible set of 
parameters.  Note that as $r_5 d_{54} << r_4$, loading
$a_5 \circ a_4$ will be done in the sparse regime.  We use the loading schedule $S=[1,2,4,3,(2,1),(2,3),(3,4), 5, (5,4)]$.  
The setup cost becomes
$r_2 d_{21}+r_2 d_{23}+r_4d_{34}+r_5d_{54} =n^{10/7}$, and the costs of loading the vertices and edges 
are all bounded by $n^{10/7}$ as given in the following tables.\\

\begin{tabular}{|r|c c c c|}
\hline
Stage & load $a_1$ & load $a_2$ & load $a_4$ & load $a_3$ \\
\hline
Global cost 
& $\sqrt{{n}/{r_1}}$ 
& $\sqrt{{n}/{r_2}}$ 
& $\sqrt{{n}/{r_4}}$ 
& $\sqrt{{n}/{r_3}}$ \\
Local cost 
& $\sqrt{n} {r_2}d_{21}/{r_1}$
& $\sqrt{n} (d_{21}+d_{23}) $ 
& $\sqrt{n} (d_{34} + r_5 d_{54}/r_4)$
& $\sqrt{n} (r_4d_{34}/r_3 + r_2 d_{23} /r_3)$ \\
\hline
Cost 
& $\sqrt{n} r_2 d_{21}/r_1$ 
& $\frac{n}{\sqrt{r_1}} (d_{21} + d_{23})$ 
& $\frac{n^{3/2}}{\sqrt{r_1r_2}} (d_{34} + r_5 d_{54}/r_4)$
& $\frac{n^{2}}{\sqrt{r_1r_2 r_4}} (r_4d_{34}/r_3 + r_2 d_{23} /r_3)$ \\
\hline
Value & $n^{13/14}$ & $n^{19/14}$ & $n^{10/7}$ & $n^{10/7}$\\
\hline
\end{tabular}

\begin{tabular}{|r|c c c|}
\hline
Stage & load $a_2 \circ a_1$ & load $a_2 \circ a_3$  & load $a_3 \circ a_4$\\
\hline
Global cost 
& $\sqrt{r_1 / d_{21}}$ 
& $\sqrt{r_3 / d_{23}}$ 
& $\sqrt{r_3 / d_{34}}$ 
\\
Local cost 
& ${r_1 }$
& ${r_3 }$
& ${r_3  }$\\
\hline Cost 
& $\frac{n^2}{\sqrt{r_2 r_3 r_4}}\sqrt{r_1}$ 
& $\frac{n^{2}}{\sqrt{r_2 r_4 d_{21}}} \sqrt{r_1r_3}$ 
& $\frac{n^{2}}{\sqrt{r_2 r_4 d_{21} d_{23}}} r_3$ 
\\
\hline
Value & $n^{10/7}$ & $n^{19/14}$ & $n^{19/14}$ \\
\hline
\end{tabular}

\begin{tabular}{|r|c c|}
\hline
Stage & load $a_5$ & load $a_5 \circ a_4$ \\
\hline
Global cost 
& $\sqrt{n/r_5}$ & \\ 
Local cost 
& $\sqrt{n} d_{54}$ & $\sqrt{r_4 r_5 }$  \\
\hline Cost 
& $\frac{n^{5/2}}{\sqrt{r_2 r_4 d_{21} d_{23} d_{34}}}\sqrt{r_3} d_{54}$ 
& $\frac{n^{5/2}}{\sqrt{r_2 d_{21} d_{23} d_{34}}} \sqrt{r_3}$  \\
\hline
Value & $n^{15/14}$ & $n^{10/7}$ \\
\hline
\end{tabular}
\end{proof}

The algorithms for finding $k$-vertex subgraphs given in~\cite{zhu11, LMS11} have complexity $O(n^{1.48})$ 
for finding a $4$-path, but it was not realized there that these algorithms apply to a much broader 
class of functions like associativity.  The key property that is used for this application is that in the basic learning 
graph model the complexity depends only on the index sets of $1$-certificates and not on the underlying alphabet.
This property was previously observed by Mario Szegedy in the context of limitations of the basic learning graph 
model \cite{Robin}.  He observed that the basic learning graph complexity of the threshold-$2$ function is 
$\Theta(n^{2/3})$, rather than the true value $\Theta(\sqrt{n})$, as threshold-2 and element distinctness have the same 
$1$-certificate index sets.

\section*{Acknowledgements}
We would like to thank Aleksandrs Belovs for discussions and comments on an earlier draft of this work.

\bibliographystyle{plain}
\bibliography{clique}
\end{document}